\documentclass[journal,twoside,web]{IEEEcolor}
\usepackage{generic}
\usepackage{cite}
\usepackage{amsmath,amssymb,amsfonts}
\usepackage{graphicx}
\usepackage{algorithm,algorithmic}
\usepackage{hyperref}
\hypersetup{hidelinks=true}
\usepackage{textcomp}
\usepackage{lipsum}
\usepackage{dsfont}
\usepackage{empheq,mathtools}

\def\BibTeX{{\rm B\kern-.05em{\sc i\kern-.025em b}\kern-.08em
    T\kern-.1667em\lower.7ex\hbox{E}\kern-.125emX}}

\newtheorem{secthm}{Theorem}[section]
\newtheorem{seccor}[secthm]{Corollary}
\newtheorem{seclem}[secthm]{Lemma}

\newtheorem{secdefn}[secthm]{Definition}
\newtheorem{secrem}[secthm]{Remark}
\newtheorem{secasm}[secthm]{Assumption}

\newcommand{\rank}{{\operatorname{rank}}}
\newcommand{\tr} { {\operatorname{trace}}}

\newcommand{\bE} { {\mathbb E}}
\newcommand{\bP} { {\mathbb P}}
\newcommand{\bR} { {\mathbb R}}
\newcommand{\bS} { {\mathbb S}}
\newcommand{\bZ} { {\mathbb Z}}

\newcommand{\fF} { {\mathbf F}}

\newcommand{\cB} { {\mathcal B}}
\newcommand{\cF} { {\mathcal F}}
\newcommand{\cJ} { {\mathcal J}}

\newcommand{\cS} { {\mathcal S}}

\newcommand{\cN} { {\mathcal N}}

\newcommand{\esssup}{\operatorname*{ess\,sup}}
\renewcommand{\theenumi}{\roman{enumi}}

\def\red{\hfill $\lhd$}

\begin{document}

\title{\LARGE \bf
Privacy protection under the exposure of systems' prior information
}

\author{
    Le Liu, Yu Kawano, \it{Member, IEEE}, and Ming Cao, \it{Fellow, IEEE}
    \\
     \thanks{The work of Liu and Cao was supported in part by the Netherlands Organization for Scientific Research (NWO-Vici-19902). The work of Liu is partly supported by Chinese Scholarship Council. The work of Yu Kawano was supported in part by JSPS KAKENHI Grant Number JP25K22804.}
\thanks{Le Liu and Ming Cao are with the Faculty of Science and Engineering, University of Groningen, 9747 AG Groningen, The Netherlands {\tt \small \{le.liu, m.cao\}@rug.nl}}
    \thanks{Yu Kawano is with the Graduate School of Advance Science and Engineering, Hiroshima University, Higashi-Hiroshima 739-8527, Japan 
        {\tt\small ykawano@hiroshima-u.ac.jp}}
}
\maketitle

\begin{abstract}
For systems whose states implicate sensitive information, their privacy is of great concern. While notions like differential privacy have been successfully introduced to dynamical systems, it is still unclear how a system's privacy can be properly protected when facing the challenging yet frequently-encountered scenario where an adversary possesses prior knowledge, e.g., the steady state, of the system. This paper presents a new systematic approach to protect the privacy of a discrete-time linear time-invariant system against adversaries knowledgeable of the system's prior information. We employ a tailored \emph{pointwise maximal leakage (PML) privacy} criterion. PML characterizes the worst-case privacy performance, which is sharply different from that of the better-known mutual-information privacy. We derive necessary and sufficient conditions for PML privacy and construct tractable design procedures. Furthermore, our analysis leads to insight into how PML privacy, differential privacy, and mutual-information privacy are related. We then revisit Kalman filters from the perspective of PML privacy and derive a lower bound on the steady-state estimation-error covariance in terms of the PML parameters. Finally, the derived results are illustrated in a case study of privacy protection for  distributed sensing in smart buildings.
\end{abstract}


\section{Introduction}

Cyber-physical systems, including smart grids \cite{yu2016smart}, autonomous mobility \cite{venancio2023cps}, and the industrial Internet of Things~\cite{sisinni2018industrial}, rely on the continuous exchange of time-series measurements for monitoring, estimation, and control. Although the increasing volume of available data enhances situational awareness and decision-making, it simultaneously increases the risk of exposing sensitive operational states to eavesdroppers or untrusted aggregators \cite{koufogiannis2017differential, han2018privacy}. This underscores the need for systematic mechanisms that can quantify and constrain what an adversary may infer from shared signals, while preserving the functional utility of cyber-physical systems.

To preserve confidentiality against privacy threats, data are often sanitized prior to transmission. Representative approaches include k-anonymity \cite{sweeney2002k}, l-diversity \cite{machanavajjhala2007diversity}, t-closeness \cite{li2006t}, and differential privacy (DP) \cite{dwork2006calibrating, dwork2006our}, along with its extensions such as concentrated DP \cite{dwork2016concentrated}, Rényi DP \cite{mironov2017renyi}, Gaussian DP \cite{dong2022gaussian}, and age DP \cite{zhang2023age}. These techniques have been successfully applied in domains such as smart grids \cite{zhao2014achieving}, health monitoring \cite{sivakumar2024addressing}, and blockchain systems \cite{zhang2019security}. However, these notions do not account for prior information and may become inefficient when such information is publicly available \cite{jiang2021context}. To address this limitation, several information-theoretic privacy notions have been proposed, including mutual-information (MI) privacy \cite{wang2016relation}, Fisher information privacy \cite{farokhi2019ensuring}, local information privacy \cite{jiang2021context}, and pointwise maximal leakage (PML) privacy \cite{saeidian2023pointwise_tit}. In this paper, we aim to tailor PML privacy to the protection of discrete-time linear time-invariant (LTI)  systems against adversaries knowledgable of prior information and to develop a systematic framework for Gaussian mechanism design.


\smallskip

\subsubsection*{Literature Review}
In dynamical systems, most privacy-preserving strategies fall into two broad categories: encryption-based schemes and noise-injection mechanisms. In the former, signals are encrypted and later recovered using secret keys. For example, \cite{altafini2019dynamical} encrypts system dynamics to protect the initial state while preserving average consensus in multi-agent networks, and the state-decomposition method in \cite{wang2019privacy} can likewise be interpreted as encryption for safeguarding initial conditions. Although such approaches can achieve strong nominal performance, their privacy guarantee depends on key secrecy and is thus vulnerable to side information \cite{hosseinalizadeh2024preserving}.

In contrast, noise-injection approaches offer robustness against side information, albeit at the expense of data quality. A canonical example is the use of decaying Laplacian noise in average consensus \cite{nozari2017differentially}, which intentionally sacrifices exact agreement on the true average to enable a tunable accuracy–privacy trade-off. Using privacy-preserving consensus, distributed optimization methods have been further developed to protect local objective function information during information sharing \cite{wang2023tailoring, wang2024robust, chen2024local}, and distributed Kalman filtering has been investigated through consensus-based schemes with noise injection for networked estimation \cite{moradi2022privacy,ke2025privacy}. 
Beyond consensus problems, privacy-preserving control has been studied through various DP mechanisms that balance control performance and privacy protection, particularly in tracking control~\cite{Yu2020,Yu2021,WKW:25,liu2024design} and linear–quadratic control~\cite{ITO2021109732,yazdani2022differentially}. In addition, noise-injection mechanisms for safeguarding the states of dynamical systems have been investigated in the contexts of initial-state protection~\cite{Yu2020,wang2023differential,liu2025initial} and differentially private Kalman filtering~\cite{le2013differentially}. However, these approaches do not explicitly account for prior information.

MI is a well known metric for quantifying privacy when prior information is available. For example, privacy-aware estimation has been formulated using MI \cite{weng2025optimal}, model randomization has been proposed to limit information leakage under MI \cite{nekouei2022model}, and privacy maximization with quantized sensor measurements has been studied via MI-based formulations \cite{murguia2021privacy}. However, because MI evaluates expected information leakage, it is not well suited to capturing rare events or worst-case scenarios. In contrast, the PML \cite{saeidian2023pointwise_tit,grosse2024extremal} provides worst-case performance guarantees with a clear operational interpretation and quantifies an adversary’s maximal information gain. Moreover, PML is robust to prior information misspecification and enjoys clean data-processing and composition properties, enabling transparent analysis and control of cumulative leakage. Despite these advantages, a PML-based framework for privacy protection in dynamical systems remains undeveloped. Existing works have primarily focused on discrete spaces \cite{saeidian2023pointwise_tit}, while \cite{grosse2024extremal} considered arbitrary probability spaces but did not exploit the structure of Gaussian distributions for deeper analytical insights.


\smallskip

\subsubsection*{Contributions}
In this paper, we aim to establish a systematic PML framework for privacy protection in discrete-time LTI Gaussian systems against adversaries with prior information, by leveraging the inherent structure of Gaussian distributions. We begin by extending the definition of PML \cite{grosse2024extremal} to the setting where (not necessarily Gaussian) random variables admit probability density functions, and by proving two fundamental properties of PML—non-negativity and a minimum value of zero—consistent with the discrete case \cite{saeidian2023pointwise_tit}.

Focusing on the static Gaussian case, we derive necessary and sufficient conditions for PML privacy by exploiting properties of the joint Gaussian distribution. These conditions enable the development of a linear matrix inequality (LMI)–based approach for designing privacy-preserving mechanisms that achieve a desired PML privacy level. Furthermore, we establish connections between PML privacy, DP, and MI privacy. More specifically, PML privacy guarantees a certain level of DP, and the converse also holds. A similar relationship is shown between PML privacy and MI privacy.

We then apply the proposed framework to a private distributed sensing problem in a network of discrete-time LTI Gaussian systems. To this end, we revisit the Kalman filter from the perspective of PML privacy, treating the system state and output as private and public information, respectively. In particular, we derive a lower bound on the steady-state estimation error covariance in terms of the PML parameters, which naturally implies that achieving a higher PML privacy level results in higher estimation error. Building on these results, we formulate the design of a private distributed sensing mechanism as a convex optimization problem that can be solved efficiently. The effectiveness of the proposed method is demonstrated through a smart-building example for multi-area climate monitoring.

The main contributions of this paper are summarized as follows:

\begin{enumerate}
  \item A comprehensive PML framework is developed in the Gaussian setting, and a necessary and sufficient condition for PML privacy is derived by leveraging the structure of the Gaussian distribution;
  \item We provide a Gaussian mechanism synthesis for privacy protection in terms of LMIs, resulting in tractable design methods;
    \item The relationships between PML privacy, DP, and MI privacy are investigated by demonstrating how the PML privacy parameters relate to those of the other frameworks.
    \item A connection between PML privacy and the Kalman filter is established by deriving an explicit lower bound on the steady-state estimation error covariance of the Kalman filter as a function of the PML privacy budget;
    \item We propose a convex optimization formulation for the privacy-aware distributed sensing problem under PML privacy constraints.
\end{enumerate}


\subsubsection*{Organization}
The remainder of this paper is organized as follows. Section~\ref{sec:MD} presents a motivating example and the definition of PML privacy.
Section~\ref{sec:static} establishes a necessary and sufficient condition for PML privacy in the Gaussian setting and introduces a Gaussian mechanism that can be designed to achieve a desired PML privacy level by solving an LMI. Moreover, relations of PML privacy with DP and MI privacy are presented. Section~\ref{sec:filter} investigates the connection between PML privacy and the Kalman filter, showing that the posterior covariance of the steady-state estimation error is lower-bounded by the PML privacy level. Section~\ref{sec:aggregation} {illustrates} the proposed results through a private distributed sensing problem.


\smallskip

\subsubsection*{Notation}
The sets of real numbers, integers and non–negative integers are denoted by $\bR$, $\bZ$ and $\bZ_{+}$, respectively. 
For $n\in\bZ_{+}$, let $\bS^{n}$ be the set of $n\times n$ real symmetric matrices, and let $\bS^{n}_{+}$ and $\bS^{n}_{++}$ denote the positive semi-definite and positive definite cones, respectively. 
For $P,Q\in\bS^{n}$, $P \succeq Q$ (resp. $P \succ Q$) means $P-Q\in\bS^{n}_{+}$ (resp. $P-Q\in\bS^{n}_{++}$). For $P \in \bS^{n}$,
$\lambda_{\min}(P) := \lambda_1(P) \le \lambda_2 (P) \le \cdots \le \lambda_n(P) =:\lambda_{\max}(P)$ denote its eigenvalues. For $A \in \bR^{n \times n}$, $\det(A)$ and $\tr(A)$ denote its determinant and trace, respectively. 
The identity matrix of dimension $n$ is denoted by $I_n$. For vectors and matrices, the $2$-norm is denoted by $|\cdot | $. The vector $2$-norm weighted by $P \succ 0$ is denoted by $|x|_P:=\sqrt{x^\top P x}$.

A probability space is denoted by $(\Omega,\cF,\bP)$, where $\Omega$ is the sample space, $\cF$ is a $\sigma$-algebra, and $\bP$ is a probability measure \cite{durrett2019probability}. For the sake of notational simplicity, an~$\bR^n$-valued random variable~$X: (\Omega, \cF) \to (\bR^n, \cB(\bR^n))$ is described by~$X: \Omega \to \bR^n$, where~$\cB(\bR^n)$ denotes the Borel $\sigma$-algebra on~$\bR^n$. Accordingly, the set of the $\bR^n$-valued random variables $X: (\Omega, \cF) \to (\bR^n, \cB(\bR^n))$ is denoted by $\cB(\Omega,\bR^n)$. For $X \in \cB (\Omega, \bR^n)$ and $Y \in \cB (\Omega,\bR^m)$, let $f_X(x)$ and $f_{X \mid Y}(x \mid y)$  denote the probability density function (PDF) of $X$ and the conditional PDF of $X$ given the observation $Y = y$, respectively. Also, $\bP_{f_X}(A)$ and $\bP_{f_{X \mid Y}}(A \mid y)$ denote the unconditional probability that $X$ is in $A$ and the conditional probability that $X$ is in $A$ given $Y=y$, respectively.
The expectation of a random variable is denoted by $\bE[\cdot]$. A Gaussian distribution with mean $\mu$ and covariance $\Sigma\succeq0$ is written as $\cN(\mu,\Sigma)$. The chi–squared distribution with $\ell$ degrees of freedom is denoted by $\chi^{2}_{\ell}$; its cumulative distribution function (CDF) is denoted by $\fF_{\chi_{\ell}}(\cdot)$, and its quantile function is $\fF^{-1}_{\chi_{\ell}}(\cdot)$.



\section{Motivation and Definitions}\label{sec:MD}
Privacy preservation is a critical concern in modern society, as system outputs may inadvertently disclose sensitive user information. A common defense strategy is to inject noise into released data. Differential privacy (DP)~\cite{dwork2006differential} provides a systematic way for noise design and privacy quantification. However, DP does not account for adversarial prior knowledge, making it difficult to leverage such information. In contrast, mutual information (MI) privacy~\cite{wang2016relation} and local information privacy~\cite{jiang2021context} explicitly exploit prior knowledge for privacy protection. Yet, MI lacks interpretability for rare events, while enforcing local information privacy is often challenging. To address these limitations, we focus on \emph{pointwise maximal leakage} (PML), which characterizes the worst-case inference by an adversary; see~\cite{saeidian2023pointwise_tit} for discrete spaces and \cite{saeidian2023pointwise} for arbitrary space. Tailoring PML to the Gaussian setting, our goal in this paper is to develop a systematic framework for the design of noise-adding privacy mechanisms.

In this section, we first introduce a real-world example to illustrate the importance of PML, and then provide its definition along with basic properties in the Gaussian setting.


\subsection{Motivating Example}
\label{sec:example}
We consider a simplified model of a smart building's temperature management system that stabilizes room temperatures around target values \cite{nekouei2022model, weng2025optimal, alisic2020ensuring}. The building is divided into $N$ zones, such as offices, meeting rooms, and residential units. The deviation of zone $i$'s temperature from its target value is described by the randam variable $X_i$, which at time $k \in \bZ_+$ fluctuates according to
\begin{align*}
    X_{i,k+1} = a_i X_{i, k} + b_i W_{i,k}, \quad i = 1,\dots, N.
\end{align*}
Here,  $W_{i,k}$ denotes the influence of human occupancy, whose randomness is  approximated by a short-memory Gaussian process, namely $W_{i,k} \sim\mathcal N(0,q_i)$ \cite{oldewurtel2012use, oldewurtel2013stochastic}; the scalar $a_i\in (-1, 1)$ governs the diminishing fluctuating temperature deviation, and $b_i \in \bR$ scales the effect of human occupancy. While the exact values of $a_i$, $b_i$ and $q_i$ are publicly available, e.g. through historical data,  the realization of $W_{i,k}$, in sharp contrast, is strictly private; for example, an office user's pattern of opening windows, or a resident's preference of sleeping times is a piece of personal information reflected in 
$W_{i,k}$. To protect privacy, the true value of the temperature deviation $X_{i,k}$ should not be transmitted directly to the temperature management system, because the transmission can be eavesdropped by internal or external adversaries. So what should be transmitted instead is the noisy version
\[
Y_{i,k} = X_{i,k} + V_{i,k},
\]
where $V_{i,k}$ denotes the additive zero-mean Gaussian noise for privacy protection. Since private $W_{i,k}$ are reflected in $X_{i,k}$ transmitted as $Y_{i,k}$ (i.e., $W_{i,k}\!\to\!X_{i,k}\!\to\!Y_{i,k}$), we take $X_{i,k}$ to be the private information to be protected. 

The noisy, thus privatized, $Y_{i,k}$ are transmitted to the temperature management center, which aggregates the received privatized information to fulfill its functionalities, e.g. to compute the overall average temperature 
\begin{align*}
    \bar Y_k &= \frac{1}{N} \sum_{i=1}^N {Y}_{i,k}.
\end{align*}
 
Since $a_i\in (-1, 1)$, the steady-state distribution of $X_{i,k}$, {denoted by $X^{i}$,} is uniquely determined by
\[
X^{i} \sim \cN(0,  \Sigma_{i, XX}), \quad \Sigma_{i, XX} = \frac{b_i^2 q_i}{1 - a_i^2}.
\]
This implies that an adversary can, among other things, deduce the steady state distribution of $X_{i,k}$ even without tapping into $Y_{i,k}$, which further highlights the more stringent requirement for any effective privacy-preserving mechanisms. In fact, to reduce the risk of privacy breach, one has to take into account the fact that an adversary has access to prior knowledge of the distributions of those random variables of interest. To handle this challenge, in what follows, we explain how we employ PML as a privacy metric and use it to develop a systematic mechanism for privacy protection.

%


\subsection{Pointwise Maximal Leakage}
Motivated by the example in the previous subsection, we investigate  PML in the Gaussian case, which yields a tractable mechanism design for privacy protection. PML was first introduced in\cite[Definition 2]{saeidian2023pointwise} and its name comes from the Pointwise Maximum discrepancy (and thus Leakage) between the prior and posterior probability of an event $A$ after observing $y$; more precisely the leakage $\ell (A \rightarrow y) $ is 
\begin{equation*}
    \log \underset{A \in \cB(\bR^n)}{\sup} 
    \frac{\bP_{f_{X \mid Y}}(A \mid {y})}{\bP_{f_X}(A)}.
\end{equation*}
To make this concept more applicable to the discussions of random variables and their distributions, we redefine it with respect to probability density functions as follows.
\begin{secdefn}\label{def:pml_metric}
    For $X \in \cB (\Omega, \bR^n)$ and $Y \in \cB (\Omega,\bR^m)$, assume that $f_{X\mid Y}(x \mid y)$ is absolutely continuous with respect to $f_X(x)$. Then, the {\emph{pointwise maximal leakage} (PML)} from $X$ to $Y = y$ is defined by
    \begin{align}
    \label{eq:PML}
        \ell (X \rightarrow y)  := \log \underset{x \in \bR^n}{\esssup} \frac{f_{X \mid Y}(x \mid y)}{f_X(x)}.
    \end{align}
 \red
\end{secdefn}

\begin{secrem}
    One notices that by integrating the probability desity function, it holds that $\ell(X \rightarrow y) \le \varepsilon$ implies $\ell(A \rightarrow y) \le \varepsilon$, and so our notation of privacy is in general stricter than that of \cite[Definition 2]{saeidian2023pointwise}. 
\red
\end{secrem}

The privacy metric \( \ell(X \rightarrow y) \) increases as the discrepancy between the posterior distribution \( f_{X \mid Y}(x \mid y) \) and the prior distribution \( f_X(X) \) grows, indicating a degradation in privacy. {From Definition~\ref{def:pml_metric}, it is expected that \( \ell(X \rightarrow y)  \) is always non-negative, and \( \ell(X \rightarrow y) = 0 \) implies maximal privacy. While this property has been established for discrete probability spaces~\cite{saeidian2023pointwise_tit}, its validity for general spaces still needs to be verified. This is what we do next. To simplify notation, we define 
$r(x) :=\frac{f_{X \mid Y}(x \mid  y)}{f_X(x)}$.

\begin{seclem}
\label{lem:PML_basic}
The PML defined in~\eqref{eq:PML} is non-negative. Moreover, for any given $Y=y$, $\ell (X \rightarrow y) = 0$ holds if and only if
    \begin{align}\label{eq:PML_basic}
         r(x) = 1 \; a.s.
    \end{align}
\end{seclem}

\begin{proof}
We first show non-negativity. 
It follows that
\begin{align}\label{pf1:PML_basic}
\bE_{f_X} [r(X)]
= \int_{\bR^n} f_{X\mid Y}(x \mid y) dx = 1.
\end{align}
This implies $\ell(X \rightarrow y) \geq 0$.
     
     We next prove the necessary and sufficient condition.
    {If \eqref{eq:PML_basic} holds, then $\underset{x \in \bR^n}{\esssup} \; r(x) =1$, yielding $\ell(X \to y) = \log 1 = 0$.} Conversely, if $\ell(X \to y) = 0$, then $\underset{x \in \bR^n}{\esssup} \; r(x) =1$, which implies $r(x) \leq 1$ almost surely. {From~\eqref{pf1:PML_basic}}, this implies $r(x) = 1$ almost surely.
\end{proof}

Lemma~\ref{lem:PML_basic} implies that a smaller value of the PML $\ell (X \rightarrow y)$ indicates that an eavesdropper gains less additional information about the {private} variable $X$ by observing $y$; see Fig.~\ref{fig:PMl_demo} illustrating how PML quantifies the discrepancy between the prior and posterior distributions. In particular when $\ell (X \rightarrow y) = 0$, an eavesdropper gains no additional information. This property is analogous to the privacy guarantee offered by DP~\cite{dwork2006differential}.

Also taking the distribution of $Y$ into account, we define $(\varepsilon, \delta)$-PML privacy as follows.
\begin{secdefn}
\label{def:pml_privacy}
 ($(\varepsilon, \delta)$-PML {privacy})
For $Y \in \cB (\Omega,\bR^m)$, we say $X \in \cB (\Omega,\bR^n)$ is  \emph{$(\varepsilon, \delta)$-PML private} if there exist $\varepsilon \geq 0$ and $\delta \in [0,1]$ such that
\begin{align}
\label{def:PML}
    \bP_{f_Y}[\ell(X \rightarrow Y) \leq \varepsilon] \geq 1-\delta
\end{align}
holds. 
\red
\end{secdefn}

\begin{figure}
    \centering
    \includegraphics[width=1\linewidth]{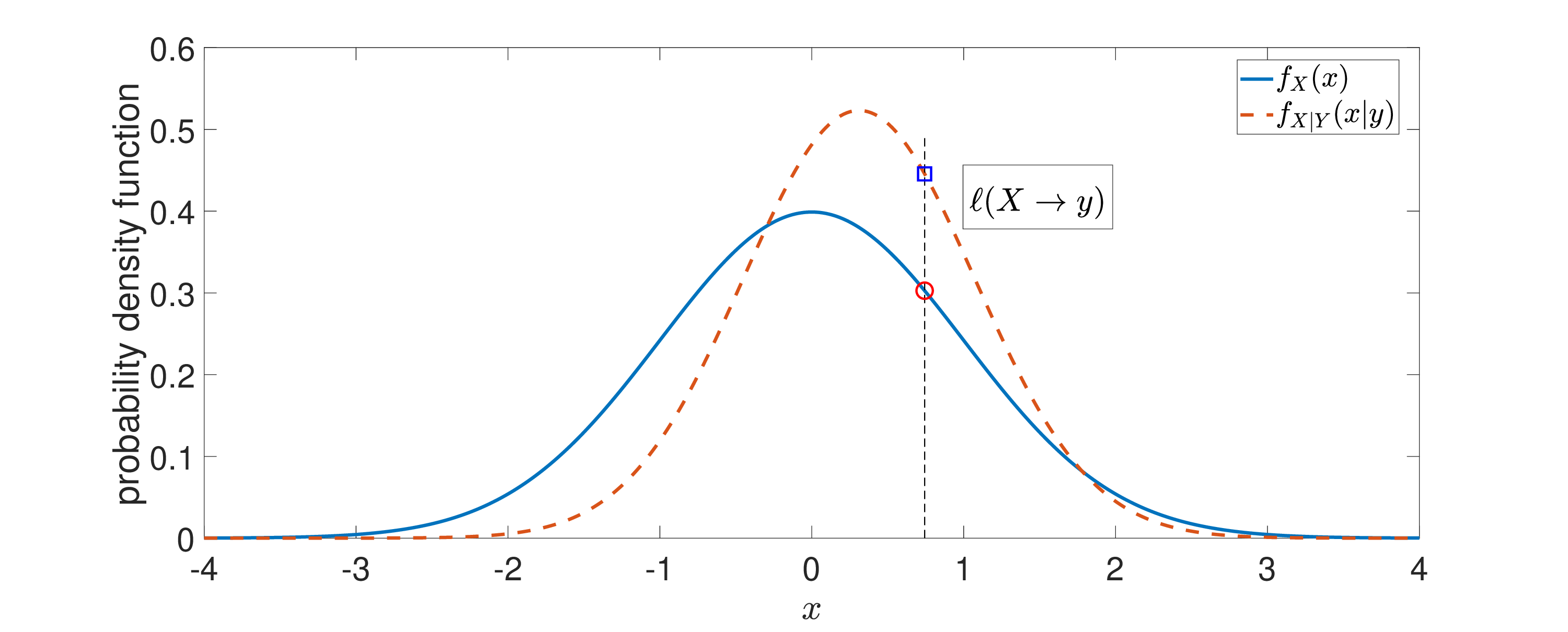}
    \caption{\textbf{Demonstration of PML}. The plot overlays the prior density $f_X(x)$ with the posterior density $f_{X\mid Y=y}(x)$ for a scalar Gaussian model.
  The dashed vertical line marks the maximizer $x^\star$ of the log-density ratio $\log\!\big(f_{X\mid Y}(x\mid y)/f_X(x)\big)$.
  At $x^\star$, the blue square and red circle indicate the posterior and prior densities, respectively.
  Intuitively, $\ell$ quantifies the pointwise amplification of belief induced by observing $Y=y$: large measurement noise yields a posterior close to the prior and a small $\ell$, whereas more informative measurements produce a sharper posterior and a larger $\ell$.}
    \label{fig:PMl_demo}
\end{figure}

    The additional parameter $\delta$ represents the probability of failure in satisfying the privacy bound $\ell(X \rightarrow y) \leq \varepsilon$, aligning with the notion of $(\varepsilon_{\rm DP}, \delta_{\rm DP})$-DP~\cite{dwork2006our}. Consequently, similar to $(\varepsilon_{\rm DP}, \delta_{\rm DP})$-DP, smaller values of $(\varepsilon, \delta)$ indicate stronger privacy guarantees by $(\varepsilon, \delta)$-PML privacy.

    In this paper, we investigate privacy preservation in linear time-invariant (LTI) systems under the $(\varepsilon,\delta)$-PML privacy framework. However, even in the static Gaussian case, mechanisms for achieving $(\varepsilon,\delta)$-PML privacy {has not been investigated before}. Therefore, we in the next section analyze the static Gaussian case, after which we extend our results to dynamical systems.



\section{Pointwise Maximal Leakage in Static Gaussian Cases}
\label{sec:static}
In this section, we first derive a necessary and sufficient condition for $(\varepsilon,\delta)$-PML privacy. We then present an LMI-based criterion to facilitate noise design for privacy protection. These results form a foundation for $(\varepsilon,\delta)$-PML privacy protection in dynamical systems.

\subsection{A Necessary and Sufficient Condition for $(\varepsilon,\delta)$-PML}
To analyze $(\varepsilon,\delta)$-PML, we assume the following property for $X \in \cB (\Omega,\bR^n)$ and $Y  \in \cB (\Omega,\bR^m)$.
\begin{secasm}
\label{asm:XY}
{Assume that $X \in \cB (\Omega,\bR^n)$ and $Y \in \cB (\Omega,\bR^m)$} follow a joint non-degenerate Gaussian distribution $\cN (\mu, \Sigma)$ with
    \begin{align}
    \label{eq:mu_sigma}
        \mu = \begin{bmatrix}
            \mu_X \\ \mu_Y
        \end{bmatrix}, \quad \Sigma = \begin{bmatrix}
            \Sigma_{XX} & \Sigma_{XY} \\ \Sigma_{XY}^\top & \Sigma_{YY}
        \end{bmatrix} \succ 0,
    \end{align}
    where $\rank(\Sigma_{XY}) = l > 0$ and $\mu_X \in \bR^n$, $\mu_Y \in \bR^m$, $\Sigma_{XX}\in \bR^{n \times n}$, $\Sigma_{XY}\in \bR^{n \times m}$, and $\Sigma_{YY}\in \bR^{m \times m}$.
    \red
\end{secasm}

This is a very mild assumption because $l > 0$ is arbitrary. If $l = 0$, then \(Y\) and \(X\) are independent. In the independent case, observing \(Y\) cannot affect the privacy of \(X\).

Under Assumption~\ref{asm:XY}, we obtain a closed-form expression of the PML $\ell(X \rightarrow y)$, which is later used to derive a necessary and sufficient condition for $(\varepsilon,\delta)$-PML privacy. This is the first main result of this paper.

\begin{secthm}
\label{thm:joint_g}
    Under Assumption~\ref{asm:XY}, the PML from $X$ to $y$ is given by
    \begin{align}
    \label{eq:jg_pml}
         &\ell(X \rightarrow y) 
    = \frac{1}{2}\log \det(\Sigma_{XX}) - \frac{1}{2}\log \det(\Gamma) + \frac{1}{2} \xi
    \end{align}
    with
    \begin{subequations}\label{eq2:jg_pml}
    \begin{align}
    \xi &:= |\Sigma_{YY}^{-1}(y-\mu_Y)|_{\Psi^{\frac{1}{2}} U U^{\top} \Psi^{\frac{1}{2}} + \Sigma_{XY}^\top\Sigma_{XX}^{-1}\Sigma_{XY}}^2, \label{eq22:jg_pml}\\
    \Psi &: =\Sigma_{YY}- \Sigma_{XY}^\top \Sigma_{XX}^{-1} \Sigma_{XY}, \textrm{ and}  \label{eq23:jg_pml}\\
    \Gamma &:= \Sigma_{XX} - \Sigma_{XY}\Sigma_{YY}^{-1}\Sigma_{XY}^\top,  \label{eq21:jg_pml}
    \end{align}
    where $U \in \bR^{m \times l}$ satisfying $U^{\top} U = I_{l}$ is defined by a compact singular value decomposition:
     \begin{align} \label{eq24:jg_pml}
    \Psi^{-\frac{1}{2}}\Sigma_{XY}^\top\Sigma_{XX}^{-1} = U D V^{\top}
    \end{align}
    \end{subequations}
    with the diagonal $D \in \bS_{++}^l$ and $V \in \bR^{n \times l}$ such that $V^{\top} V = I_l$.    
\end{secthm}
\begin{proof}
    See Appendix \ref{app:thm:joint_g}.
\end{proof} 

From~\eqref{eq:jg_pml}, if we consider $\ell(X \rightarrow Y)$ with random variable $Y$, only $\xi$ becomes a random variable. In fact, it is possible to show that $\xi$ follows an $l$-freedom $\chi^2$ distribution. Consequently, we derive the following necessary and sufficient condition for \((\varepsilon, \delta)\)-PML privacy as the second main result of this paper.

\begin{secthm}
\label{thm:privacy_con}
   For $X$ and $Y$ satisfying Assumption~\ref{asm:XY},  $X$ is $(\varepsilon, \delta)$-PML private if and only if 
    \begin{align}
    \label{eq:privacy_con}
        \fF^{-1}_{\chi^2_l}(1-\delta) \leq 2\varepsilon - \log \det(\Sigma_{XX}) + \log \det(\Gamma).
    \end{align}
\end{secthm}
\begin{proof}
        See Appendix~\ref{app:privacy_con}.
    \end{proof}

Since prior information, in particular \( \Sigma_{XX} \) is given, the PML privacy level $(\varepsilon, \delta)$ can be improved by adjusting $\Gamma$ in order to increase \( \log \det(\Gamma)  \). This observation is consistent with the Bayesian estimation theory, where \( \log \det(\Gamma) \) serves as a measure of estimation uncertainty~\cite{Boyd2008}. When \( \log \det(\Gamma) \) approaches to \( \log \det(\Sigma_{XX}) \), the observation \( Y \) contributes little to reducing uncertainty, indicating less privacy leakage.

It is important to note that the PML privacy level \((\varepsilon, \delta)\) cannot be specified arbitrarily, as \( - \log \det(\Gamma) + \log \det(\Sigma_{XX}) > 0 \) must hold. Consequently, we have the following necessary condition for \((\varepsilon, \delta)\)-PML privacy.
\begin{seccor}\label{cor:privacy_con}
For $X$ and $Y$ satisfying Assumption~\ref{asm:XY}, if $X$ is $(\varepsilon, \delta)$-PML private, then
\begin{align}\label{asm:nece}
\frac{1}{2} \fF^{-1}_{\chi^2_l}(1 - \delta) < \varepsilon.
\end{align}
\end{seccor}

Although Theorem~\ref{thm:privacy_con} provides a necessary and sufficient conditions for \((\varepsilon, \delta)\)-PML privacy, it does not offer a systematic method for designing noise to achieve \((\varepsilon, \delta)\)-PML privacy. Therefore, in the next subsection, we develop such mechanisms.

\subsection{Gaussian Mechanisms for Privacy Preservation}
In this subsection, we consider designing Gaussian noise to ensure \((\varepsilon, \delta)\)-PML privacy. Specifically, we develop an LMI-based approach to design the following Gaussian mechanism.

\begin{secdefn}
Let $X \in \cB (\Omega,\bR^n)$ and $Z \in \cB (\Omega,\bR^m)$ follow a joint non-degenerate Gaussian distribution $\cN (\mu, \Sigma)$ with
    \begin{align*}
        \mu = \begin{bmatrix}
            \mu_X \\ \mu_Z
        \end{bmatrix} \quad \textrm{and}\quad
        \Sigma = \begin{bmatrix}
            \Sigma_{XX} & \Sigma_{XZ} \\ \Sigma_{XZ}^\top & \Sigma_{ZZ}
        \end{bmatrix} \succ 0,
    \end{align*} 
    where $\rank(\Sigma_{XZ}) = l > 0$.
    Also, let $V \in \cB (\Omega,\bR^m)$ follow a Gaussian distribution $\cN(0, \Theta)$.
     Then, the masked output
    \begin{align}
        \label{eq:Gaussian_mech}
        Y = Z + V
    \end{align}
    is called a \emph{Gaussian mechanism}. Furthermore, the Gaussian mechanism is said to be $(\varepsilon, \delta)$-PML private if there exist $\varepsilon \geq 0$ and $\delta \in (0, 1)$ such that~\eqref{def:PML} and~\eqref{asm:nece} hold. \red
\end{secdefn}

As discussed previously, \eqref{asm:nece} is a necessary condition for \((\varepsilon, \delta)\)-PML privacy, which is easy to verify. Building upon this necessary condition, we next aim to derive a sufficient condition that ensures $(\varepsilon, \delta)$-PML privacy.

To simplify the presentation of the forthcoming result, we introduce the following notation:
\begin{align}
    \label{eq:kappa}
    \kappa_{\ell}(\varepsilon, \delta) := \exp \left(\frac{ \fF^{-1}_{\chi^2_l}(1 - \delta) - 2\varepsilon}{n}\right).
\end{align}
Using this definition, we now state the third main result of this section.
\begin{secthm}
    \label{thm:pml_g}
    A Gaussian mechanism~\eqref{eq:Gaussian_mech} is $(\varepsilon, \delta)$-PML private if $\Theta$ is designed such that
    \begin{align}
    \label{eq:lmi_con}
        \begin{bmatrix}
            \left(1-\kappa_{\ell}(\varepsilon, \delta)\right) \Sigma_{XX} & \Sigma_{XZ} \\ \Sigma_{XZ}^\top  & \Theta + \Sigma_{ZZ}
        \end{bmatrix} \succeq 0.
    \end{align}
\end{secthm}
\begin{proof}
    See Appendix~\ref{app:pml_g}.
\end{proof}

In fact, the necessary condition \eqref{asm:nece} is critical, since it guarantees that the top-left block in \eqref{eq:lmi_con} is positive definite. From Theorem~\ref{thm:pml_g}, one can directly compute the covariance matrix $\Theta$ for a Gaussian mechanism \eqref{eq:Gaussian_mech} via the LMI formulation, enhancing its practical applicability. Applying the Schur Complement~\cite[Theorem 7.7.6]{horn2012matrix} in view of ~\eqref{asm:nece}, one can check that \eqref{eq:lmi_con} is equivalent to
\begin{align}\label{eq2:lmi_con}
\Theta + \Sigma_{ZZ} - \frac{1}{1-\kappa_{\ell}(\varepsilon, \delta)} \Sigma_{XZ}^{\top} \Sigma_{XX}^{-1} \Sigma_{XZ} \succeq 0.
\end{align}
Thus, selecting $\Theta$ sufficiently large always ensures \eqref{eq:lmi_con}.
Here, we want to underscore the consequent observation that it is \emph{always} feasible to design an $(\varepsilon, \delta)$-PML private Gaussian mechanism.

An important special case of the Gaussian mechanism~\eqref{eq:Gaussian_mech} arises when \( Z \) is a linear function of \( X \), i.e., \( Z = C X \) with \( C \in \bR^{m \times n} \). Accordingly, we consider the following Gaussian mechanism:
\begin{align}
\label{eq:Gaussian}
    Y = C X + V.
\end{align}
For clarity, we call this mechanism  the \emph{linear Gaussian mechanism}. 
The linear Gaussian mechanism naturally appears in a variety of privacy problems, including database algorithm design~\cite{balle2018improving} and initial state privacy of linear dynamical systems~\cite{wang2023differential, Yu2020}. In fact, differential measures such as DP are employed in these literature.

As a corollary of Theorem~\ref{thm:pml_g}, we derive the following LMI condition for $(\varepsilon, \delta)$-PML privacy of the linear Gaussian mechanism.

\begin{seccor}
\label{prop:linear}
    A linear Gaussian mechanism \eqref{eq:Gaussian} is $(\varepsilon, \delta)$-PML private if $\Theta$ is designed such that 
    \begin{align}
        \label{eq:lmi_con_linear2}
       \Theta \succeq \frac{\kappa_{\ell}(\varepsilon, \delta)}{1 - \kappa_{\ell}(\varepsilon, \delta)}C\Sigma_{XX}C^{\top}.
    \end{align}
    
\end{seccor}
\begin{proof}
From $Z = C X$, we have $\Sigma_{XZ} = \Sigma_{XX} C^T$ and $\Sigma_{ZZ} = C \Sigma_{XX} C^T$. Substituting this into~\eqref{eq2:lmi_con} gives~\eqref{eq:lmi_con_linear2}.
\end{proof}

Given fixed \(\varepsilon \ge 0\) and \(\delta \in (0, 1)\), $\kappa_{\ell}(\varepsilon,\delta)/(1 - \kappa_{\ell}(\varepsilon, \delta))$ increases as the parameter \(l \in \bZ_+\setminus\{0\}\) grows, requiring larger $\Theta$. Note that $\ell = {\rm rank} (C)$ for $Z = C X$. Therefore, if $Z = C X$ contains more information about $X$ (e.g., $\ell' > \ell$), we need to add larger noise to guarantee the same privacy performance for $\kappa_{\ell'}(\varepsilon,\delta)$.

From the necessary condition~\eqref{asm:nece}, we have \(\kappa_{\ell}(\varepsilon, \delta) \in (0,1)\). Thus, if the covariance \(\Sigma_{XX}\) is large, we need to design large \(\Theta\). This implies that when prior information is limited (i.e., \(\Sigma_{XX}\) is large), large noise is required to be added for achieving the same level of privacy protection.


\subsection{Relationships with Other Privacy Notions}
In this paper, we employ $(\varepsilon, \delta)$-PML privacy as a privacy criterion of a noise-adding mechanism. Other well-known criteria are DP~\cite{balle2018improving} and MI privacy~\cite{wang2016relation}. In this subsection, we study the relations of $(\varepsilon, \delta)$-PML privacy with respect to these two privacy criteria.

First, we investigate the relation between $(\varepsilon, \delta)$-PML privacy and DP defined below.
\begin{secdefn} \label{def:dp_privacy}
(Differential Privacy \cite{balle2018improving})
Given $\zeta>0$, let $\operatorname{Adj}^\zeta$ denote the set of all pairs of initial states $(x, x') \in \bR^{n} \times \bR^{n} $ satisfying $ | x - x'|_{2} \leq \zeta$. Then, the mechanism 
\begin{align}\label{eq:dp_mech}
 M(x) := C x + V, \quad  V \in \cB (\Omega,\bR^m)
\end{align}
is said to be \emph{$(\varepsilon_{\rm DP}, \delta_{\rm DP})$-differentially private (DP)} for $\operatorname{Adj}^\zeta$ if there exist $\varepsilon_{\rm DP}, \delta_{\rm DP} \geq 0$  such that 
\begin{align}
\bP \left(M(x) \in \cS\right) \leq \mathrm{e}^{\varepsilon_{\rm DP}} \bP \left(M(x') \in \cS\right)+\delta_{\rm DP}, 
\quad \forall \cS \in \cB(\bR^m)
\end{align}
for any $(x, x') \in \mathrm{Adj}^\zeta$.
\red
\end{secdefn}

The following theorem characterizes the relationship between PML privacy and DP. If a linear Gaussian mechanism~\eqref{eq:Gaussian} is \((\varepsilon, \delta)\)-PML private, the corresponding mechanism~\eqref{eq:dp_mech} is $(\varepsilon_{\rm DP}, \delta_{\rm DP})$-differentially private for a suitable pair $(\varepsilon_{\rm DP}, \delta_{\rm DP})$, and vice versa.

\begin{secthm}
\label{thm:pmldp}
Consider a linear Gaussian mechanism~\eqref{eq:Gaussian} with the $\operatorname{rank}(C) = l$, and the corresponding mechanism~\eqref{eq:dp_mech}. Also, recall $\kappa_{\ell}(\varepsilon, \delta)$ in \eqref{eq:kappa}. Then, the following two hold:
\begin{enumerate}
\item If \eqref{eq:Gaussian} is $(\varepsilon, \delta)$-PML private, then~\eqref{eq:dp_mech} is $(\varepsilon_{\rm DP}, \delta_{\rm DP})$-differentially private for $\operatorname{Adj}^\zeta$ for any $\varepsilon_{\rm DP} > 0$, $\delta_{\rm DP} \in (0, 1)$ and $\zeta > 0$ satisfying
\begin{align}
    \label{eq:pmltodp}
\frac{1}{ (\kappa_{\ell}(\varepsilon, \delta))^n \lambda_{\min} (\Sigma_{XX})}
\leq \left(\frac{\phi^{-1}_{\varepsilon_{\rm DP}}(\delta_{\rm DP})} {\zeta}\right)^2,
\end{align}

    where $\phi_{\varepsilon_{\rm DP}}^{-1}(\delta_{\rm DP})$ denotes the inverse function of $\phi$ with respect to $\delta_{\rm DP}$ for fixed $\varepsilon_{\rm DP}$, defined by
\begin{align*}
\phi(\varepsilon_{\rm DP}, \delta_{\rm DP}) &:= \varphi\left(\frac{\delta_{\rm DP}}{2} - \frac{\varepsilon_{\rm DP}}{\delta_{\rm DP}}\right) \\
&\qquad - e^{\varepsilon_{\rm DP}} \varphi \left(-\frac{\delta_{\rm DP}}{2} - \frac{\varepsilon_{\rm DP}}{\delta_{\rm DP}}\right)\\
\varphi (w) &:= \frac{1}{\sqrt{2\pi}} \int_{-\infty}^w e^{-\frac{v^2}{2}} \, dv;
\end{align*}
\item If~\eqref{eq:dp_mech} is $(\varepsilon_{\rm DP}, \delta_{\rm DP})$-differentially private for $\operatorname{Adj}^\zeta$, then \eqref{eq:Gaussian} is $(\varepsilon, \delta)$-PML private for any $\varepsilon \ge 0$ and $\delta \in (0, 1)$ such that~\eqref{asm:nece} and 
\begin{align}
    \label{eq:dptopml}
            &n \log \kappa_{\ell}(\varepsilon, \delta) \nonumber\\
            &\le  - \log \det \left(I_n +  \left(\frac{\phi^{-1}_{\varepsilon_{\rm DP}}(\delta_{\rm DP})}{\zeta}\right)^2 \Sigma_{XX} \right) 
    \end{align}
    hold.
\end{enumerate}
\end{secthm}

\begin{proof}
    See Appendix~\ref{app:pmldp}.
\end{proof}

    In Theorem~\ref{thm:pmldp}, the prior covariance $\Sigma_{XX}$ plays a key role in connecting $(\varepsilon,\delta)$-PML privacy with $(\varepsilon_{\rm DP},\delta_{\rm DP})$-differential privacy. The function $\phi^{-1}_{\varepsilon_{\rm DP}}(\delta_{\rm DP})$ increases in both $\delta_{\rm DP}$ and $\varepsilon_{\rm DP}$ \cite{balle2018improving}. Thus, for a fixed $(\varepsilon, \delta)$-PML privacy level, \eqref{eq:pmltodp} indicates that a larger $\Sigma_{XX}$ yields smaller $(\varepsilon_{\rm DP}, \delta_{\rm DP})$, implying stronger differential privacy. This agrees with our earlier observation after Corollary~\ref{prop:linear} that a larger $\Sigma_{XX}$ requires injecting larger noise to achieve the same $(\varepsilon, \delta)$-PML privacy level. In general, injecting larger noise increases privacy levels, regardless of the criterion used.  
    
    Next, from~\eqref{eq:kappa}, $\kappa_{\ell}(\varepsilon, \delta)$
     decreases in both $\delta$ and $\varepsilon$. Thus, for a fixed $(\varepsilon_{\rm DP}, \delta_{\rm DP})$-DP level, \eqref{eq:dptopml} shows that a larger $\Sigma_{XX}$ results in a weaker $(\varepsilon, \delta)$-PML privacy guarantee. Since $(\varepsilon_{\rm DP}, \delta_{\rm DP})$ and $\Sigma_{XX}$ are  independent, when $\Sigma_{XX}$ becomes sufficiently large, DP becomes less effective to evaluate privacy performance with prior information. This illustrates the importance of $(\varepsilon, \delta)$-PML privacy when prior information is available.

Next, we study the relation between $(\varepsilon, \delta)$-PML privacy and MI privacy as defined below.
\begin{secdefn}
     (Mutual-Information Privacy \cite{wang2016relation}) The linear Gaussian mechanism~\eqref{eq:Gaussian} is $\varepsilon_{\rm MI}$-mutual-information (MI) private  if there exists  $\varepsilon_{\rm MI} > 0$ such that
\begin{align}\label{eq:mip}
    I(X; Y)  \leq \varepsilon_{\rm MI},
\end{align}
where $I(X;Y)$ is the MI between $X$ and $Y$.~\red
\end{secdefn}

The next theorem states the relationship between PML privacy and MI privacy. 

\begin{secthm}
\label{thm:pmlmi}
Consider a linear Gaussian mechanism~\eqref{eq:Gaussian} with $\operatorname{rank}(C) = l$. Also, recall $\kappa_{\ell}(\varepsilon, \delta)$ in \eqref{eq:kappa}. Then, the following two hold:

\begin{enumerate}
\item If \eqref{eq:Gaussian} is $(\varepsilon, \delta)$-PML private, then it is $\varepsilon_{\rm MI}$-MI private for any $\varepsilon_{\rm MI} > 0$ satisfying 
\begin{align}
\label{eq:mi2}
 n \log (\kappa_{\ell}(\varepsilon,\delta)) \ge -2 \varepsilon_{\rm MI}  ;
\end{align}
\item If \eqref{eq:Gaussian} is $\varepsilon_{\rm MI}$-MI private, then it is $(\varepsilon, \delta)$-PML private for any $\varepsilon \ge 0$ and $\delta \in [0, 1]$ such that~\eqref{asm:nece} and 
    \begin{align}
    \label{eq:mi}
         n \log(\kappa_{\ell}(\varepsilon, \delta)) \le - 2\varepsilon_{\rm MI}
    \end{align}
    hold.
\end{enumerate}
\end{secthm}
\begin{proof}
    See Appendix~\ref{app:pmlmi}.
\end{proof}

    Item~i) of Theorem~\ref{thm:pmlmi} yields a direct bound on the $\varepsilon_{\rm MI}$-MI privacy parameter from $(\varepsilon,\delta)$-PML privacy. In contrast, item~ii) shows that $\varepsilon_{\rm MI}$-MI privacy provides an indirect bound on $(\varepsilon,\delta)$-PML privacy via $\kappa_{\ell}(\varepsilon,\delta)$. Using~\eqref{eq:kappa}, \eqref{eq:mi} can be rearranged as
\[
\varepsilon_{\rm MI} + \frac{1}{2} \fF_{\chi^2_l}^{-1}(1-\delta) \le \varepsilon.
\]
This implies that even when $\varepsilon_{\rm MI}$ is small, reducing the probability that the worst-case leakage exceeds a given threshold, i.e., making $\delta$ small, causes the increase of $\varepsilon$.
This reflects the fundamental distinction between MI and PML: the former evaluates average privacy leakage, whereas the later is concerned with the worst-case privacy leakage.



\section{Private Kalman Filtering}
\label{sec:filter}
In this section, we revisit the Kalman filter from the viewpoint of \((\varepsilon, \delta)\)-PML privacy, where the state and output represent private and public information, respectively. We estimate a lower bound on the covariance of the steady-state estimation error by using \((\varepsilon, \delta)\). 

Consider a discrete-time LTI Gaussian system, described by
\begin{align}\label{eq:LTI}
    \left\{
    \begin{alignedat}{2}
        X_{k+1} &= A X_k + W_k,\\
        Y_k &= C X_k + V_k,
    \end{alignedat}\right.
\end{align}
where the state $X_k \in \cB (\Omega,\bR^n)$ and output $Y_k \in \cB (\Omega,\bR^m)$ are private and public information, respectively. 
Different from standard estimation problems, the state noise $W_k \sim \cN(0, Q)$ and measurement noise $V_k \sim \cN(0, \Theta)$ are designed and added in the purpose of privacy protection, where $W_k$ and $V_l$, $k,l \in \bZ_+$ are independent from each other, and each $W_k$ and $V_k$ is independent across time $k \in \bZ_+$. Matrices $Q \in \bS_{++}^n$ and $\Theta \in \bS_{++}^m$ as well as Schur stable $A \in \bR^{n \times n}$ and $C \in \bR^{m \times n}$ of rank $m$ are public. {It is standard to assume that $C$ is of full row rank; however, the results in this section can be readily extended to the case of ${\rm rank} \; C = \ell < m$, $\ell > 0$.}

We consider a scenario where an eavesdropper computes the state estimate \( \hat X_k \) using the Kalman filter \cite{anderson2005optimal}:
\begin{align}\label{eq:Kalman}
\hat X_{k+1} = A \hat X_k + K_k (Y_k - C A \hat X_k),
\end{align}
where
\begin{subequations}\label{eq:pk}
\begin{align}
    P^-_k &= A P_{k-1} A^{\top} + Q, \\
    K_k &= P^-_k C^\top (C P^-_k C^\top + \Theta )^{-1}, \\
    P_k &= (I_n - K_k C ) P^-_k.
\end{align}
\end{subequations}

Since $A$ is Schur stable, the pair $(C,A)$ is detectable. Thus, $Q \succ 0$ and $\Theta \succ 0$ imply that the following steady-state covariance is symmetric and positive definite:
\begin{align}\label{eq:P}
P:= \lim_{k \to \infty} P_k = \lim_{k \to \infty} \bE [(\hat X_k - X_k)(\hat X_k -X_k)^\top].
\end{align}
We use this as a privacy metric of the state and establish a connection with PML-privacy. In particular, we estimate a lower bound on $P$ by the PML parameters $(\varepsilon, \delta)$ of the linear Gaussian mechanism~\eqref{eq:Gaussian} with the same $C$ and $Q$ as those in the system~\eqref{eq:LTI}, where $l = m$.  We take the steady-state distribution of $X_k$ to be the prior information of~\eqref{eq:Gaussian}; here, $\Sigma_{XX} \in \bS_{++}^n$ is the solution to the following discrete-time Lyapunov equation:
\begin{align}\label{eq:Lya}
    \Sigma_{XX} = A \Sigma_{XX} A^{\top} + Q.
\end{align}
From the Schur stability of $A$, we know that $\Sigma_{XX} \in \bS_{++}^n$ is the unique symmetric and positive definite solution.

The following theorem provides a lower bound on $P$ by using the PML parameters $(\varepsilon, \delta)$. 
\begin{secthm}\label{thm:kalman_lb}
Consider the LTI system~\eqref{eq:LTI} and its Kalman filter~\eqref{eq:Kalman}. If the corresponding linear Gaussian mechanism~\eqref{eq:Gaussian} is $(\varepsilon, \delta)$-PML private with the prior variance $\Sigma_{XX}$ in~\eqref{eq:Lya}, then it follows that
\begin{align}\label{eq:PMLsys}
   \log \det(P) \ge  \fF^{-1}_{\chi^2_m}(1-\delta) - 2\varepsilon + \log \det(Q),
\end{align}
where $P \in \bS_{++}^n$ is the steady-state covariance~\eqref{eq:P}.
\end{secthm}
\begin{proof}
See Appendix~\ref{app:kalman}
\end{proof}

Theorem~\ref{thm:kalman_lb} give a lower bound on $\det(P)$. We can also obtain a lower bound on $\tr(P)$ as follows.

\begin{seccor}\label{cor:kalman_lb}
Consider the  LTI system~\eqref{eq:LTI} and its Kalman filter~\eqref{eq:Kalman}. If the corresponding linear Gaussian mechanism~\eqref{eq:Gaussian} is $(\varepsilon, \delta)$-PML private with the prior variance $\Sigma_{XX}$ in~\eqref{eq:Lya}, then it follows that
\begin{align}\label{eq2:PMLsys}
    \tr(P) \geq n + \fF^{-1}_{\chi^2_m}(1-\delta) - 2\varepsilon + \log \det(Q),
\end{align}
where $P \in \bS_{++}^n$ is the steady-state covariance~\eqref{eq:P}.
\end{seccor}

\begin{proof}
It follows from $a - 1 \ge \log(a)$, $a > 0$ that
\begin{align*}
\tr(P) - n 
&= \sum_{i=1}^n (\lambda_i(P) - 1) \\
&\ge \sum_{i=1}^n \log(\lambda_i(P)) = \log \det(P).
\end{align*}
This and~\eqref{eq:PMLsys} imply~\eqref{eq2:PMLsys}.
\end{proof}

Theorem~\ref{thm:kalman_lb} and Corollary~\ref{cor:kalman_lb} provide lower bounds on the steady-steady covariance error $P$ using the parameters $(\varepsilon, \delta)$ of PML privacy. These lower bounds correspond to natural observations: i) increasing $\det (Q)$ degrades estimation accuracy for fixed $(\varepsilon, \delta)$, i.e., injecting large state noise enhances state privacy; ii) from Corollary~\ref{prop:linear}, for fixed $\Sigma_{XX}$, enlarging $\Theta$ tends to increasing the $(\varepsilon, \delta)$-PML privacy level, which degrades estimation accuracy, i.e., enhances state privacy. We further demonstrate how our results can be applied in the following section.



\section{Distributed Privacy-Aware Aggregation}
\label{sec:aggregation}
Motivated by the example in Section~\ref{sec:example}, we apply our results to a problem of distributed privacy-aware measurement aggregation. Our objective is to design privacy-preserving noise for each local measurement when an adversary employs a Kalman filter for state estimation. We employ convex optimization techniques to design the optimal additive noise for each local measurement, balancing the trade-off between privacy guarantee and aggregation accuracy, which is illustrated by the example in Section~\ref{sec:example}.

\subsection{Optimal Noise Design}
Consider a distributed system, consisting of $N$ discrete-time LTI Gaussian subsystems:
\begin{align}\label{eq:sys-dynamics}
    \left\{
    \begin{alignedat}{2}
        X_{i, k+1} &= A_i X_{i, k} + W_{i,k}, \\
        Y_{i, k} &= C_i X_{i,k} + V_{i,k}, \quad i = 1, \ldots, N,
    \end{alignedat}\right.
\end{align}
where the local state $X_{i,k} \in \cB (\Omega,\bR^{n_i})$ and local output $Y_{i,k} \in \cB (\Omega,\bR^{m_i})$ are private and public information, respectively. The input noise ${W}_{i,k}\sim \cN(0, Q_i)$ and $V_{i,k} \sim \cN(0, \Theta_i)$ are designed for the purpose of privacy protection, where they are independent of each other and also independent among each agent.
For each $i=1,\dots,N$, $Q_i \in \bS_{++}^{n_i}$, $\Theta_i \in \bS_{++}^{n_i}$, Schur stable $A_i \in \bR^{n_i \times n_i}$, and $C_i \in \bR^{m_i \times n_i}$ of rank $m_i$ are public.

Each subsystem transmits its privatized output \({Y}_{i,k}\) to the fusion center, which then aggregates all its received measurements to compute the global measurement:
\begin{align}\label{eq:aggY}
    \bar{Y}_k = \sum_{i=1}^N L_i {Y}_{i,k},
\end{align}
where $\bar{Y}_k \in \bR^q$, and \(L_i \in \bR^{q \times m_i}\) denotes the aggregation weight associated with subsystem \(i\).

Noise ${W}_{i,k}$ and $V_{i,k}$ added for privacy protection in general degenerate data equality. To quantify this effect, we define the following accuracy metric:
\begin{align*}
    \cJ
    &:= \sum_{i=1}^N \bE [ (L_i ({Y}_{i,k} - C_i X_{i,k}))^\top L_i ({Y}_{i,k} - C_i X_{i,k})] \\
    &\; = \sum_{i=1}^N \tr(L_i \Theta_i L_i^\top).
\end{align*}
A smaller \(\mathcal{J}\) implies a smaller aggregation error, i.e., higher data accuracy.
Our objective is to design the noise that minimizes $\mathcal{J}$ while guaranteeing a required privacy level for each local plant $i=1,\dots,N$. 

We use $(\varepsilon_i, \delta_i)$-PML privacy of the corresponding linear Gaussian mechanism~\eqref{eq:Gaussian} as the privacy criterion of each local plant. As shown in Theorem~\ref{thm:kalman_lb}, $(\varepsilon_i, \delta_i)$ gives a lower bound on the steady-state error covariance of the Kalman filter if the prior distribution is selected as the solution $\Sigma_{i,XX} \in \bS_{++}^{n_i}$ to the following discrete-time Lyapunov equation:
\begin{align}\label{eq:Lya_l}
    \Sigma_{i,XX} = A_i \Sigma_{i,XX} A_i^{\top} + Q_i.
\end{align}
Then, from Corollary~\ref{prop:linear}, the linear Gaussian mechanism~\eqref{eq:Gaussian} corresponding to each local plant is $(\varepsilon_i, \delta_i)$-PML private if $\Theta_i$ is designed as to satisfy
\begin{align}\label{eq:lmi_con_linear}
       \Theta_i \succeq \frac{\kappa_{\ell_i}(\varepsilon_i, \delta_i)}{1 - \kappa_{\ell_i}(\varepsilon_i, \delta_i)}C_i\Sigma_{i,XX}C_i^{\top},
\end{align}
where $\kappa_{\ell}(\varepsilon, \delta)$ is defined in~\eqref{eq:kappa}.

In summary, the distributed privacy-constrained aggregation problem with input noise can be formulated as the following distributed optimization problem:
\begin{align}\label{eq:opt1}
&\min_{{Q}_i, \Theta_i, \Sigma_{i,XX}} \tr (L_i \Theta_i L_i^{\top}) \\
& \text{s.t.} \quad {Q}_i \succ 0, \Theta_i \succ 0, \text{\eqref{eq:Lya_l}, \eqref{eq:lmi_con_linear}} \nonumber.
\end{align}
This contains an equality constraint~\eqref{eq:Lya_l}, which can be relaxed into
\begin{align}\label{eq:opt2}
&\min_{\hat Q_i, \Theta_i, \Sigma_{i,XX}} \tr (L_i \Theta_i L_i^{\top}) \\
& \text{s.t.} \quad \text{$\hat Q_i \succ 0, \Theta_i \succ 0$, $\Sigma_{i,XX} \succeq A_i \Sigma_{i,XX} A_i^{\top} + \hat Q_i$, \eqref{eq:lmi_con_linear}} \nonumber.
\end{align}
Selecting $Q_i := \Sigma_{i,XX} - A_i \Sigma_{i,XX} A_i^{\top} (\succeq \hat Q_i)$ recovers the solution to~\eqref{eq:opt1}. Replacing $\hat Q_i, \Theta_i \succ 0$ with $\hat Q_i \succeq c I_n$ and $\Theta_i \succeq c I_m$ for $c>0$ makes~\eqref{eq:opt2} a convex semi-definite program (SDP). 

Since the distributed privacy-aware measurement aggregation problem is formulated as a convex optimization problem~\eqref{eq:opt2}, it is easy to include additional constraints and design conditions. For instance, one may influence an eavesdropper's prior knowledge by specifying the range of the  prior distribution, e.g. through adding $\overline{\Sigma}_{i,XX} \succeq \Sigma_{i,XX} \succeq \underline{\Sigma}_{i,XX}$, where $\overline{\Sigma}_{i,XX}$ and $\underline{\Sigma}_{i,XX}$ are upper and lower bounds, respectively. Similarly, bounds on the distributions of stage and measurement noise can be included.

\begin{secrem}
\label{rem:noise_design}
In the example in Section~\ref{sec:example}, the distribution $\cN(0, Q_i)$ of $W_{i,k}$ is not a design parameter as $W_{i,k}$ models human behaviors. In this case, solving~\eqref{eq:Lya_l} gives $\Sigma_{i,XX}$. Moreover, from~\eqref{eq:lmi_con_linear}, the optimal solution $\Theta_i$ is an arbitrary $\Theta_i \succ 0$ such that
\begin{align}\label{eq:opt3}
L_i \Theta_i L_i^\top 
= L_i \frac{\kappa_{\ell_i}(\varepsilon_i, \delta_i)}{1 - \kappa_{\ell_i}(\varepsilon_i, \delta_i)}C_i\Sigma_{i,XX}C_i^{\top} L_i^\top.
\end{align}
 \red
\end{secrem}


\subsection{Simulations}
We revisit the motivating example in Section~\ref{sec:example} with the parameters in~\cite{weng2025optimal}. {Consider \eqref{eq:sys-dynamics} with $N = 3$}, where $A_i = 0.75$, $C_i = 1$ and $L_i = 1/3$, {$i=1,2,3$}. 
{In this problem, $Q_i = 0.4$, $i=1,2,3$ is not a design parameter. Also from~\eqref{eq:Lya_l}, we have $\Sigma_{i,XX}= Q_i/(1 -A_i^2) = 0.9143$, $i=1,2,3$.} We select the privacy parameters as $\varepsilon_1 = 6$, $\varepsilon_2 = 7$, $\varepsilon_3 = 8$, and $\delta_i = 0.001$, $i = 1,2,3$, which satisfies~\eqref{asm:nece}. 

{Each optimal solution $\Theta_i$ to~\eqref{eq:opt1} is obtained by solving~\eqref{eq:opt3} as in}
\begin{align*}
    \Theta_1 = 0.410, \; \Theta_2 = 0.0400, \; \Theta_3= 0.0052.
\end{align*}
This verifies that stronger $(\varepsilon_i, \delta_i)$-PML privacy level (in this case $i=1$) yields larger variance $\Theta_i$. In other words, larger noise is required to be injected. Accordingly, the estimation performance by the Kalman filter becomes poor as confirmed by Fig.~\ref{fig:1}. Finally, Fig.~\ref{fig:4} presents the true aggregated values $\bar{Z}_k = \sum_{i=1}^N L_i C_i X_{i,k}$ and privatized one~\eqref{eq:aggY}, showing the inherent privacy–accuracy trade-off.

\begin{figure}
    \centering
    \includegraphics[width=1\linewidth]{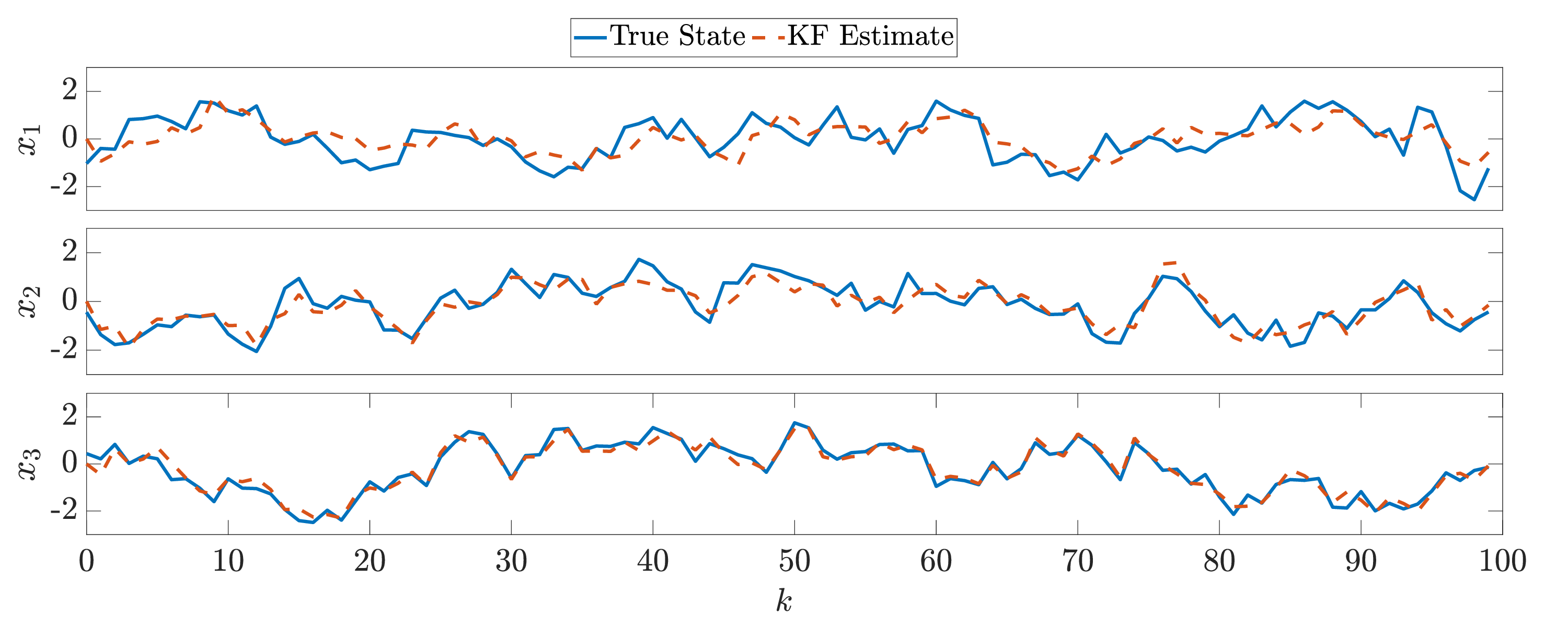}
    \caption{True State and its Kalman filter estimate for each subsystem}
    \label{fig:1}
\end{figure}

\begin{figure}
    \centering
    \includegraphics[width=1\linewidth]{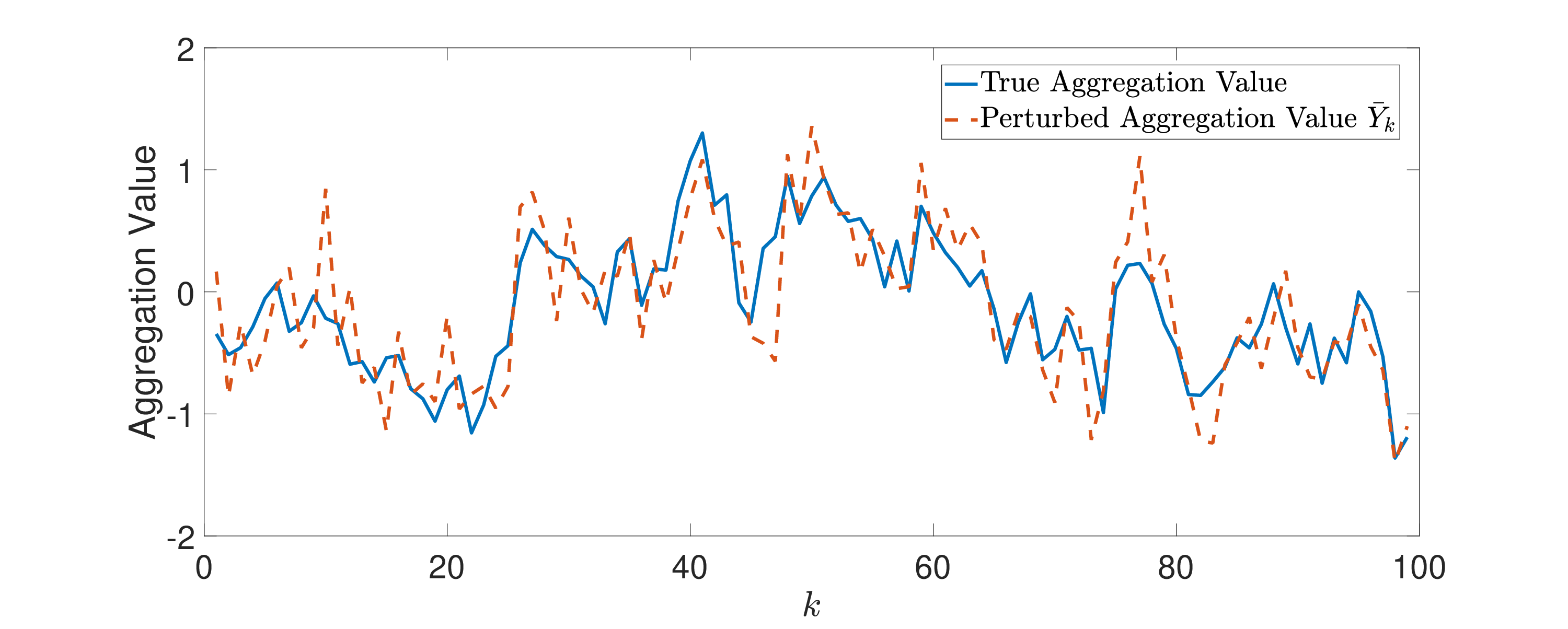}
    \caption{Aggregation Value Comparison}
    \label{fig:4}
\end{figure}

\section{Conclusion}\label{sec:con}
In this paper, we have investigated PML in Gaussian settings. For static Gaussian models, we have established a necessary and sufficient condition for PML and have proposed a framework of Gaussian mechanism design in terms of LMIs. We have also shown a relation of PML with DP and MI privacy. Furthermore, in the context of discrete-time linear Gaussion systems, we have studied the connection between Kalman filters and PML. In particular, we have shown that PML gives a lower bound on the covariance of steady-state estimation error. Finally, we have applied the proposed results to a distributed measurement-aggregation architecture in which each subsystem transmits privatized measurements prior to aggregation and have validated the resulting privacy–accuracy trade-off on multi-area climate-monitoring case studies. 

Future work includes extending our results to a more general class of systems, such as nonlinear and uncertain systems, as well as considering different types of noise, including heavy-tailed and correlated noise. Moreover, co-design of control and noise for closed-loop privacy protection remains an open problem.



\appendices
\section{Proof of Theorem \ref{thm:joint_g}}
\label{app:thm:joint_g}
(Step 1) We reformulate the supremum problem as a standard least squares problem and compute its maximizer $x^*$. 
Since the joint Gaussian distribution is non-degenerate, $\Gamma$ is of full rank. 
Thus, it follows from the standard Bayesian analysis~\cite{anderson2005optimal} that 
\begin{align}\label{pf1:joint_g}
    &f_{X \mid Y}(x \mid y) \nonumber\\
    &  = \frac{1}{(2\pi)^{\frac{n}{2}}\sqrt{\det (\Gamma)}}  \nonumber \\
    & \times \exp \left(-\frac{1}{2} |x - \mu_X - \Sigma_{XY} \Sigma_{YY}^{-1} (y - \mu_Y)|_{\Gamma^{-1}}^2 \right).
\end{align}
Utilizing~\eqref{pf1:joint_g}, it can be computed that 
    \begin{align}\label{pf2:joint_g}
    \frac{f_{X \mid Y}(x \mid y)}{f_X(x)}
=  \frac{\sqrt{\det (\Sigma_{XX})}}{\sqrt{\det (\Gamma)}} \exp \Bigl(-\frac{1}{2} \Pi \Bigr),
\end{align}
where
\begin{align*}
\Pi &:= |x - \mu_X - \Sigma_{XY} \Sigma_{YY}^{-1} (y - \mu_Y)|_{\Gamma^{-1}}^2 - |x - \mu_X|_{\Sigma_{XX}^{-1}}^2 \nonumber\\
&\; = |x - \mu_X |_{\Gamma^{-1} - \Sigma_{XX}^{-1}}^2 \nonumber\\
&\qquad - 2(y- \mu_{Y})^{\top}\Sigma_{YY}^{-1}\Sigma_{XY}^\top\Gamma^{-1}(x - \mu_{X}) \nonumber\\
&\qquad + |\Sigma_{XY} \Sigma_{YY}^{-1} (y - \mu_Y)|_{\Gamma^{-1}}^2.
\end{align*}
By the matrix inversion lemma~\cite[Theorem 3.2.2]{petersen2008matrix}, we obtain
\begin{align}
\label{eq:Gamma}
    \Gamma^{-1} 
    &= (\Sigma_{XX} - \Sigma_{XY}\Sigma_{YY}^{-1}\Sigma_{XY}^\top)^{-1} \nonumber \\
    & = \Sigma_{XX}^{-1} + \Sigma_{XX}^{-1} \Sigma_{XY}\Psi^{-1}\Sigma_{XY}^\top \Sigma_{XX}^{-1}.
\end{align}
Consequently, the matrix in the first term satisfies
\begin{align*}
\Gamma^{-1} - \Sigma_{XX}^{-1} = \Sigma_{XX}^{-1} \Sigma_{XY}\Psi^{-1}\Sigma_{XY}^\top \Sigma_{XX}^{-1}.
\end{align*}
From~\eqref{eq23:jg_pml} and~\eqref{eq:Gamma}, the matrix in the second term can be rearranged as
\begin{align*}
&\Sigma_{YY}^{-1} \Sigma_{XY}^\top\Gamma^{-1} \\
&=\Sigma_{YY}^{-1} \Sigma_{XY}^\top(\Sigma_{XX}^{-1} + \Sigma_{XX}^{-1} \Sigma_{XY}\Psi^{-1}\Sigma_{XY}^\top \Sigma_{XX}^{-1}) \\
&= \Sigma_{YY}^{-1}\left(\Psi + \Sigma_{XY}^\top\Sigma_{XX}^{-1}\Sigma_{XY}\right) \Psi^{-1}\Sigma_{XY}^\top \Sigma_{XX}^{-1} \\
&= \Psi^{-1}\Sigma_{XY}^\top \Sigma_{XX}^{-1}.
\end{align*}
Substituting these into~$\Pi$ yields
\begin{align}\label{pf3:joint_g}
\Pi 
& = |\Sigma_{XY}^\top \Sigma_{XX}^{-1} (x - \mu_X) |_{\Psi^{-1}}^2 \nonumber\\
&\qquad - 2(y- \mu_{Y})^{\top}\Psi^{-1}\Sigma_{XY}^\top \Sigma_{XX}^{-1}(x - \mu_{X}) \nonumber\\
&\qquad + |\Sigma_{XY} \Sigma_{YY}^{-1} (y - \mu_Y)|_{\Gamma^{-1}}^2 \nonumber\\
& = |\Sigma_{XY}^\top \Sigma_{XX}^{-1} (x - \mu_X) - (y- \mu_{Y})|_{\Psi^{-1}}^2 \nonumber\\
&\qquad + |\Sigma_{XY} \Sigma_{YY}^{-1} (y - \mu_Y)|_{\Gamma^{-1}}^2
  - |(y - \mu_Y)|_{\Psi^{-1}}^2.
\end{align}
Since \( \Sigma \) is positive definite, by the Schur complement, \( \Psi \) in~\eqref{eq23:jg_pml} is also positive definite. Therefore, by standard least squares analysis, it follows from~\eqref{pf2:joint_g} and~\eqref{pf3:joint_g} that 
\begin{align*}
        x^{*}& =\underset{x \in \bR^n}{\arg \esssup} \frac{f_{X \mid Y}(x \mid Y = y)}{f_X(x)} \\
         &=  \underset{x \in \bR^n}{\arg \max} \frac{f_{X \mid Y}(x \mid Y = y)}{f_X(x)}\\
         &= \underset{x \in \bR^n}{\arg \min} |\Sigma_{XY}^\top \Sigma_{XX}^{-1} (x - \mu_X) - (y- \mu_{Y})|_{\Psi^{-1}}^2.
\end{align*}

From~\eqref{eq24:jg_pml}, $x^* \in \bR^n$ is a minimizer of $|\Sigma_{XY}^\top \Sigma_{XX}^{-1} (x - \mu_X) - (y- \mu_{Y})|_{\Psi^{-1}}^2$ if and only if
\begin{align*}
U D V^{\top} (x^* - \mu_X) - \Psi^{-\frac{1}{2}}(y- \mu_{Y}) = 0.
\end{align*}
According to \cite{golub2013matrix} and the structure of the compact SVD, its standard solution is 
\begin{align}\label{pf4:joint_g}
x^{*} = \mu_X + V D^{-1} U^\top \Psi^{-\frac{1}{2}}(y - \mu_Y).
\end{align}

(Step 2)
We derive the expression~\eqref{eq:jg_pml}.
Substituting $x^*$ of~\eqref{pf4:joint_g} into \eqref{pf3:joint_g} leads to 
\begin{align*}
\Pi  & = |(U U^\top - I) \Psi^{-\frac{1}{2}} (y - \mu_Y)|^2 \\
&\qquad + |\Sigma_{XY} \Sigma_{YY}^{-1} (y - \mu_Y)|_{\Gamma^{-1}}^2 - |(y - \mu_Y)|_{\Psi^{-1}}^2 \\
&= -|U^\top \Psi^{-\frac{1}{2}} (y - \mu_Y)|^2 + |\Sigma_{XY} \Sigma_{YY}^{-1} (y - \mu_Y)|_{\Gamma^{-1}}^2\\
&= -|\Sigma_{YY}^{-1}(y - \mu_Y)|_{\Sigma_{YY} \Psi^{-\frac{1}{2}} U U^\top \Psi^{-\frac{1}{2}} \Sigma_{YY} 
- \Sigma_{XY}^\top \Gamma^{-1}\Sigma_{XY}}^2.
\end{align*}
From~\eqref{eq:PML} and~\eqref{pf2:joint_g}, to establish \eqref{eq:jg_pml}, it remains to show
\begin{align}\label{pf5:joint_g}
&\Sigma_{YY} \Psi^{-\frac{1}{2}} U U^\top \Psi^{-\frac{1}{2}} \Sigma_{YY} - \Sigma_{XY}^\top \Gamma^{-1}\Sigma_{XY} \nonumber\\
&= \Psi^{\frac{1}{2}} U U^{\top} \Psi^{\frac{1}{2}} + \Sigma_{XY}^\top\Sigma_{XX}^{-1}\Sigma_{XY}.
\end{align}

{
To show this, define
\begin{align}\label{pf6:joint_g}
\Lambda : = \Sigma_{XY}^\top \Sigma_{XX}^{-1}\Sigma_{XY}. 
\end{align}
Then, we have 
\begin{align}\label{pf7:joint_g}
\Sigma_{YY} = \Psi + \Lambda.
\end{align}
Using this, $\Sigma_{YY}\Psi^{-\frac{1}{2}}UU^{\top}\Psi^{-\frac{1}{2}}\Sigma_{YY}$ can be rearranged as
\begin{align*}
    &\Sigma_{YY}\Psi^{-\frac{1}{2}}UU^{\top}\Psi^{-\frac{1}{2}}\Sigma_{YY} \\
    &= (\Psi + \Lambda) \Psi^{-\frac{1}{2}} U U^{\top} \Psi^{-\frac{1}{2}}(\Psi + \Lambda) \\
    &=  \Psi^{\frac{1}{2}}U U^{\top}\Psi^{\frac{1}{2}} + \Lambda \Psi^{-\frac{1}{2}} U U^{\top} \Psi^{-\frac{1}{2}}\Lambda +  \Lambda \Psi^{-\frac{1}{2}} U U^{\top} \Psi^{\frac{1}{2}}\\
    &\qquad +  \Psi^{\frac{1}{2}} U U^{\top} \Psi^{-\frac{1}{2}} \Lambda.
\end{align*}
 Using~\eqref{eq24:jg_pml} and~\eqref{pf6:joint_g}, it can be computed that
\begin{align*}
\Psi^{\frac{1}{2}} U U^{\top} \Psi^{-\frac{1}{2}} \Lambda 
&= \Psi^{\frac{1}{2}} U U^{\top} \Psi^{-\frac{1}{2}} \Sigma_{XY}^\top \Sigma_{XX}^{-1}\Sigma_{XY}\\
&= \Psi^{\frac{1}{2}} U D V^\top \Sigma_{XY}\\
&= \Sigma_{XY}^\top\Sigma_{XX}^{-1} \Sigma_{XY} = \Lambda.
\end{align*}
Also, left multiplying both sides by $\Lambda \Psi^{-1}$ gives
\begin{align*} 
     \Lambda \Psi^{-\frac{1}{2}} U U^{\top} \Psi^{-\frac{1}{2}}\Lambda  = \Lambda  \Psi^{-1} \Lambda .
\end{align*}
Thus, we have
\begin{align}\label{pf8:joint_g}
    &\Sigma_{YY}\Psi^{-\frac{1}{2}}UU^{\top}\Psi^{-\frac{1}{2}}\Sigma_{YY} \nonumber\\
    &=  \Psi^{\frac{1}{2}}U U^{\top}\Psi^{\frac{1}{2}} + \Lambda \Psi^{-1} \Lambda + 2 \Lambda.
\end{align}

Next, using~\eqref{eq:Gamma} and \eqref{pf6:joint_g},it can be shown that
\begin{align}\label{pf9:joint_g}
    &\Sigma_{XY}^\top \Gamma^{-1} \Sigma_{XY} \nonumber\\
    &= \Sigma_{XY}^\top (\Sigma_{XX}^{-1} + \Sigma_{XX}^{-1} \Sigma_{XY}\Psi^{-1}\Sigma_{XY}^\top \Sigma_{XX}^{-1}) \Sigma_{XY}  \nonumber\\
    &=  \Lambda + \Lambda \Psi^{-1} \Lambda.
\end{align}
In summary, \eqref{pf6:joint_g}--\eqref{pf9:joint_g} imply \eqref{pf5:joint_g}.
\QED}


\section{Proof of Theorem \ref{thm:privacy_con}}
\label{app:privacy_con}
(Step 1) By abusing the notation \( \xi \), we first show that \( 
{\xi} \) follows a $\chi^2$ distribution with \( l \) degrees of freedom, i.e., \( \xi \) can be expressed as the sum of squares of \( l \) independent standard normal random variables.

    Let $z = \Sigma_{YY}^{-\frac{1}{2}}(y - \mu_Y)$. Then we have $\bE[zz^{\top}] = I_{m}$. It follows from~\eqref{eq22:jg_pml} that
    \begin{align}\label{pf1:lem1}
        &{\xi} = z^{\top} {\Xi} z\\
        &\quad {\Xi} := \Sigma_{YY}^{-\frac{1}{2}}( \Psi^{\frac{1}{2}}U U^{\top}\Psi^{\frac{1}{2}} + \Sigma_{XY}^\top \Sigma_{XX}^{-1}\Sigma_{XY} )\Sigma_{YY}^{-\frac{1}{2}}. \nonumber
    \end{align}
 
    Then, We show ${\Xi}^2 = {\Xi}$. From~\eqref{eq24:jg_pml}, \eqref{pf6:joint_g}, and~\eqref{pf7:joint_g}, we have
\begin{align*}
\Sigma_{XY}^\top \Sigma_{XX}^{-1}\Sigma_{XY} = \Psi^{\frac{1}{2}} U \hat D  U^\top \Psi^{\frac{1}{2}},
\end{align*}
and
\begin{align*}
\Sigma_{YY}^{-1} 
&= ( \Psi  + \Sigma_{XY}^\top \Sigma_{XX}^{-1}\Sigma_{XY})^{-1} \\
        & =\Psi^{-\frac{1}{2}} ( I_m  +  U \hat D U^\top )^{-1} \Psi^{-\frac{1}{2}},
\end{align*}
where $\hat D := DV^{\top} \Sigma_{XX} V D \succ 0$.
Thus, ${\Xi}$ can be rearranged as
\begin{align}\label{pf2:lem1}
{\Xi}
&= \Sigma_{YY}^{-\frac{1}{2}}( \Psi^{\frac{1}{2}}U U^{\top}\Psi^{\frac{1}{2}} + \Psi^{\frac{1}{2}} U \hat D  U^\top \Psi^{\frac{1}{2}})\Sigma_{YY}^{-\frac{1}{2}} \nonumber\\
&= \Sigma_{YY}^{-\frac{1}{2}}\Psi^{\frac{1}{2}}U(  I_l + \hat D)U^\top \Psi^{\frac{1}{2}}\Sigma_{YY}^{-\frac{1}{2}},
\end{align}
and consequently, ${\Xi}^2$ is
\begin{align*}
{\Xi}^2
&= \Sigma_{YY}^{-\frac{1}{2}}\Psi^{\frac{1}{2}}U(  I_l + \hat D) U^\top ( I_m  +  U\hat D U^\top )^{-1} U\\
&\qquad \times (  I_l + \hat D)U^\top \Psi^{\frac{1}{2}}\Sigma_{YY}^{-\frac{1}{2}}.
\end{align*}
We further simplify this. By the matrix inversion lemma, we obtain
\begin{align*}
 ( I_m  +  U\hat D U^\top )^{-1} = I_m - U (\hat D^{-1} + I_l)^{-1} U^\top.
\end{align*}
Multiplying $U^\top$ from left and $U$ from right yields
\begin{align*}
U^\top ( I_m  +  U\hat D U^\top )^{-1} U 
&= I_l - (\hat D^{-1} + I_l)^{-1}\\
&= I_l - \hat D^{-1}(I_l + \hat D)^{-1}\\
&=(I_l + \hat D)^{-1}.
\end{align*}
Therefore, we have
\begin{align*}
{\Xi}^2
&= \Sigma_{YY}^{-\frac{1}{2}}\Psi^{\frac{1}{2}}U(  I_l + \hat D) (I_l + \hat D)^{-1}(  I_l + \hat D)U^\top \Psi^{\frac{1}{2}}\Sigma_{YY}^{-\frac{1}{2}}\\
&=\Sigma_{YY}^{-\frac{1}{2}}\Psi^{\frac{1}{2}}U(I_l + \hat D)U^\top \Psi^{\frac{1}{2}}\Sigma_{YY}^{-\frac{1}{2}} 
= {\Xi},
\end{align*}
where~\eqref{pf2:lem1} is used in the last equality.

Since ${\Xi}$ is an idempotent matrix, its eigenvalues are either $0$ or $1$ \cite[Problem 5, p.37]{horn2012matrix}. Also, from~\eqref{pf2:lem1}, 
$\rank({\Xi}) = \rank(\hat D) = l$. Therefore, from~\eqref{pf1:lem1}, \( \xi \) follows a $\chi^2$ distribution with \( l \) degrees of freedom.

(Step 2)
From~\eqref{eq:jg_pml}, \eqref{def:PML} is equivalent to
\begin{align*}
    \bP_{f_Y}[\log \det(\Sigma_{XX}) - \log \det(\Gamma) + \xi \leq 2\varepsilon] \geq 1-\delta.
\end{align*}
    Since ${\xi}$ follows an $l$-freedom $\chi^2$ distribution, the above inequality is equivalent to
    \begin{align*}
        \fF_{\chi^2_l}\left(2\varepsilon - \log \det(\Sigma_{XX}) + \log \det(\Gamma) \right) \geq 1 - \delta.
    \end{align*}
    Since the cumulative distribution function of the $\chi^2$ distribution is continuous and strictly increasing, it always admits the inverse~\cite[Theorem 5.6.5]{bartle2000introduction}, it follows that  
    \begin{align*}
        \fF^{-1}_{\chi^2_l}(1-\delta) \le 2\varepsilon - \log \det(\Sigma_{XX}) + \log \det(\Gamma)
    \end{align*}
   which is exactly \eqref{eq:privacy_con}. 
    \QED


\section{Proof of Theorem~\ref{thm:pml_g}}
\label{app:pml_g}

For the Gaussian mechanism~\eqref{eq:Gaussian_mech}, we have $\Sigma_{YY} = \Sigma_{ZZ} + \Theta$ and $\Sigma_{XY} = \Sigma_{XZ}$. 
According to Theorem~\ref{thm:privacy_con}, the Gaussian mechanism~\eqref{eq:Gaussian_mech} is $(\varepsilon, \delta)$-PML private if and only if
\begin{align*}
   & \fF^{-1}_{\chi^2_l}(1-\delta) \leq 2\varepsilon - \log \det(\Sigma_{XX}) + \log \det(\Gamma) \\
   &\qquad \Gamma := \Sigma_{XX} - \Sigma_{XZ} (\Sigma_{ZZ} + \Theta)^{-1} \Sigma_{XZ}^\top.
\end{align*}
It is further equivalent to
\begin{align*}
        \fF^{-1}_{\chi^2_l}(1-\delta) - 2\varepsilon 
        &\leq  
        \log \det(\Gamma \Sigma_{XX}^{-1})\\
        &= \log \det(\Sigma_{XX}^{-\frac{1}{2}}\Gamma \Sigma_{XX}^{-\frac{1}{2}}).
\end{align*}
Taking the exponential of both sides gives
\begin{align}\label{pf1:pml_g}
    \exp \left(\fF^{-1}_{\chi^2_l}(1-\delta) - 2\varepsilon \right)
    \leq \det(\Sigma_{XX}^{-\frac{1}{2}}\Gamma \Sigma_{XX}^{-\frac{1}{2}}).
\end{align}
It can be computed that
\begin{align*}
    \Sigma_{XX}^{-\frac{1}{2}}\Gamma\Sigma_{XX}^{-\frac{1}{2}} 
    =I_{n} - \Sigma_{XX}^{-\frac{1}{2}}\Sigma_{XZ} (\Sigma_{ZZ} + \Theta)^{-1} \Sigma_{XZ}^\top \Sigma_{XX}^{-\frac{1}{2}}.
\end{align*}
Using~\eqref{eq:kappa}, \eqref{pf1:pml_g} can be rearranged as
\begin{align}\label{pf2:pml_g}
(\kappa_{\ell}(\varepsilon, \delta))^n
    \leq \det(I_{n} - \Sigma_{XX}^{-\frac{1}{2}}\Sigma_{XZ} (\Sigma_{ZZ} + \Theta)^{-1} \Sigma_{XZ}^\top \Sigma_{XX}^{-\frac{1}{2}}).
\end{align}
Thus, it suffices to show that~\eqref{eq:lmi_con} implies~\eqref{pf2:pml_g}.

Applying the Schur complement to~\eqref{eq:lmi_con} yields
\begin{align*}
    (1-\kappa_{\ell}(\varepsilon, \delta))\Sigma_{XX} - \Sigma_{XZ} (\Sigma_{ZZ} + \Theta)^{-1} \Sigma_{XZ}^\top  \succeq 0.
\end{align*}
Multiplying $\Sigma_{XX}^{-\frac{1}{2}}$ from both side gives
\begin{align*}
    (1-\kappa_{\ell}(\varepsilon, \delta))I_n - \Sigma_{XX}^{-\frac{1}{2}}\Sigma_{XZ} (\Sigma_{ZZ} + \Theta)^{-1} \Sigma_{XZ}^\top \Sigma_{XX}^{-\frac{1}{2}} \succeq 0,
\end{align*}
{i.e.,}
\begin{align*}
    \kappa_{\ell}(\varepsilon, \delta)I_n
    \preceq I_n - \Sigma_{XX}^{-\frac{1}{2}}\Sigma_{XZ} (\Sigma_{ZZ} + \Theta)^{-1} \Sigma_{XZ}^\top \Sigma_{XX}^{-\frac{1}{2}}.
\end{align*}
This implies~\eqref{pf2:pml_g}.
\QED


\section{Proof of Theorem~\ref{thm:pmldp}}
\label{app:pmldp}

The proof is established based on the following lemma.
\begin{seclem}
\label{lem:dp}
    For any $\varepsilon_{\rm DP} \geq 0$, $\delta_{\rm DP} \in (0,1)$, and $\zeta >0$, \eqref{eq:dp_mech} is $(\varepsilon_{\rm DP}, \delta_{\rm DP})$-differentially private for $\operatorname{Adj}^\zeta$ if and only if 
    \begin{align}\label{eq:dp}
     \lambda_{\max} (C^\top \Theta^{-1} C) \leq \left( \frac{\phi^{-1}_{\varepsilon_{\rm DP}}(\delta_{\rm DP})}{\zeta} \right)^2
    \end{align}
    holds.
\end{seclem}

\begin{proof}
When $\Theta = I_m$, a necessary and sufficient condition for $(\varepsilon_{\rm DP}, \delta_{\rm DP})$-differentially privacy is $\phi(\varepsilon_{\rm DP}, \zeta |C|) \leq \delta_{\rm DP}$ \cite[Lemma 1]{wang2023differential}. For non-necessarily identity $\Theta \succ 0$, a necessary and sufficient condition becomes  $\phi(\varepsilon_{\rm DP}, \zeta |\Theta^{-\frac{1}{2}} C|) \leq \delta_{\rm DP}$. This is equivalent to $| \Theta^{-\frac{1}{2}} C| \leq \phi^{-1}_{\varepsilon_{\rm DP}}(\delta_{\rm DP})/ \zeta$.
\end{proof}

Now, we are ready to prove Theorem~\ref{thm:pmldp}.

(Proof of (i)) 
According to Theorem~\ref{thm:privacy_con}, the Gaussian mechanism~\eqref{eq:Gaussian} is $(\varepsilon, \delta)$-PML private if and only if~\eqref{pf1:pml_g} holds.
Using~\eqref{eq:kappa} and $\Gamma^{-1} = \Sigma_{XX}^{-1} + C^{\top} \Theta^{-1} C$ following from~\eqref{eq:Gamma}, \eqref{pf1:pml_g} can be rearranged as
\begin{align}\label{pf1:privacy_con}
   &\frac{1}{ (\kappa_{\ell}(\varepsilon, \delta))^n \det(\Sigma_{XX})} \nonumber\\
   &\ge \det (\Gamma^{-1}) = \det (\Sigma_{XX}^{-1} + C^{\top} \Theta^{-1} C).
\end{align}

From this, we consider obtaining an upper bound on $|\Theta^{-\frac{1}{2}} C|^2 =\lambda_{\max} (C^{\top} \Theta^{-1} C)$. It follows that
\begin{align*}
&\det (\Sigma_{XX}^{-1} + C^{\top} \Theta^{-1} C) \\
&= \prod_{i =1}^n \lambda_i (\Sigma_{XX}^{-1} + C^{\top} \Theta^{-1} C)\\
&\ge \lambda_{\max} (\Sigma_{XX}^{-1} + C^{\top} \Theta^{-1} C) \prod_{i =1}^{n-1} \lambda_i (\Sigma_{XX}^{-1}).
\end{align*}
By Weyl's inequality~\cite[Thorem 4.3.1]{horn2012matrix}, we have
\begin{align}\label{pf2:privacy_con}
&\det (\Sigma_{XX}^{-1} + C^{\top} \Theta^{-1} C) \nonumber\\
&\ge (\lambda_{\max} (C^{\top} \Theta^{-1} C) + \lambda_{\min} (\Sigma_{XX}^{-1})) \prod_{i =1}^{n-1} \lambda_i (\Sigma_{XX}^{-1}) \nonumber\\
&= \frac{\lambda_{\max} (C^{\top} \Theta^{-1} C) + \lambda_{\min} (\Sigma_{XX}^{-1})}{\lambda_{\max} (\Sigma_{XX}^{-1}) \det(\Sigma_{XX})}.
\end{align}
Therefore, \eqref{pf1:privacy_con} and \eqref{pf2:privacy_con} lead to
\begin{align*}
\lambda_{\max} (C^{\top} \Theta^{-1} C) 
&\le \frac{\lambda_{\max} (\Sigma_{XX}^{-1})}{ (\kappa_{\ell}(\varepsilon, \delta))^n} - \lambda_{\min} (\Sigma_{XX}^{-1})\\
&= \frac{1}{ (\kappa_{\ell}(\varepsilon, \delta))^n \lambda_{\min} (\Sigma_{XX})} - \frac{1}{\lambda_{\max} (\Sigma_{XX})}.
\end{align*}
If~\eqref{eq:pmltodp} holds, we have \eqref{eq:dp}, and consequently, \eqref{eq:dp_mech} is $(\varepsilon_{\rm DP}, \delta_{\rm DP})$-differentially private for $\operatorname{Adj}^\zeta$ by Lemma~\ref{lem:dp}.

(Proof of (ii))
According to Lemma~\ref{lem:dp}, \eqref{eq:dp_mech} is $(\varepsilon_{\rm DP}, \delta_{\rm DP})$-differentially private for $\operatorname{Adj}^\zeta$ if and only if \eqref{eq:dp} holds. It follows from~$\Gamma^{-1} = \Sigma_{XX}^{-1} + C^{\top} \Theta^{-1} C$ and~\eqref{eq:dp} that
\begin{align*}
    \det (\Gamma^{-1})
    &=\det (\Sigma_{XX}^{-1} + C^{\top} \Theta^{-1} C)\\
    &\leq \det \left( \Sigma_{XX}^{-1} + \left(\frac{\phi^{-1}_{\varepsilon_{\rm DP}}(\delta_{\rm DP})}{ \zeta}\right)^2 I_n\right),
\end{align*}
or equivalently,
\begin{align*}
    \log \det (\Gamma) \geq - \log \det \left( \Sigma_{XX}^{-1} + \left(\frac{\phi^{-1}_{\varepsilon_{\rm DP}}(\delta_{\rm DP})}{ \zeta}\right)^2 I_n\right).
\end{align*}
This, together with inequality~\eqref{eq:dptopml}, yields
\begin{align*}
    \log \det (\Gamma) - n \log \kappa_{\ell}(\varepsilon, \delta) -\log\det (\Sigma_{XX}) \geq 0.
\end{align*}
From Theorem~\ref{thm:privacy_con} and~\eqref{eq:kappa}, this implies that~\eqref{eq:Gaussian} is $(\varepsilon, \delta)$-PML private.
\QED


\section{Proof of Theorem~\ref{thm:pmlmi}}
\label{app:pmlmi}

The proof is established based on the following lemma.
\begin{seclem}
\label{lem:mi}
    A linear Gaussian mechanism~\eqref{eq:Gaussian} with $\operatorname{rank}(C) = l$ is $(\varepsilon, \delta)$-PML private if and only if 
\begin{align}
\label{eq:mi_pml}
   - 2I(X;Y) \ge n \log (\kappa_{\ell}(\varepsilon,\delta)).
\end{align} 
holds.   
\end{seclem}
\begin{proof}
By definition, the MI between \( X \) and \( Y \) is
\[
I(X; Y) = h(Y) - h(Y \mid X),
\]
where \( h(\cdot) \) denotes the differential entropy. Since \( X \sim \cN(\mu_X, \Sigma_{XX}) \) and \( N \sim \cN(0, \Theta) \) are independent, we have
\[
Y \sim \cN(C \mu_X, C \Sigma_{XX} C^\top + \Theta).
\]
Using the entropy formula for a multivariate Gaussian distribution~\cite[Theorem 8.4.1]{cover1999elements}, 
we obtain
\begin{align*}
h(Y) &= \frac{1}{2} \log \left( (2\pi e)^m \det(C \Sigma_{XX} C^\top + \Theta) \right) \\
h(Y \mid X) &= \frac{1}{2} \log \left( (2\pi e)^m \det(\Theta) \right).
\end{align*}
Thus, we obtain
\begin{align*}
I(X; Y) 
&= \frac{1}{2} \log \left( \frac{\det(C \Sigma_{XX} C^\top + \Theta)}{\det(\Theta)} \right)\\
&=\frac{1}{2} \log \det(I_m + \Theta^{-\frac{1}{2}} C \Sigma_{XX} C^\top \Theta^{-\frac{1}{2}}).
\end{align*}
Applying the Sylvester's determinant identity~\cite[Problem 17, p.57]{horn2012matrix} yields
\begin{align*}
I(X; Y) 
&= \frac{1}{2} \log \det \left( I_n + \Sigma_{XX}^{1/2} C^\top \Theta^{-1} C \Sigma_{XX}^{1/2} \right) \\
&= \frac{1}{2} \left( \log \det (\Sigma_{XX}) + \log \det ( \Sigma_{XX}^{-1} +  C^\top \Theta^{-1} C ) \right).
\end{align*}

From the proof of Theorem~\ref{thm:pmldp}, \eqref{eq:Gaussian} is $(\varepsilon, \delta)$-PML private if and only if \eqref{pf1:privacy_con}, or equivalently,
\begin{align*}
   - n\log (\kappa_{\ell}(\varepsilon, \delta))
   &\ge \log\det(\Sigma_{XX}) + \\
   & \log \det (\Sigma_{XX}^{-1} + C^{\top} \Theta^{-1} C)
    = 2 I(X; Y) 
\end{align*}
holds. This is nothing but~\eqref{eq:mi_pml}.
\end{proof}

Now, we are ready to prove Theorem~\ref{thm:pmlmi}. 

(Proof of (i))
From Lemma~\ref{lem:mi}, \eqref{eq:Gaussian} is $(\varepsilon,\delta)$-PML private if and only if \eqref{eq:mi_pml} holds. It follows from~\eqref{eq:mi2} and \eqref{eq:mi_pml} that
\begin{align*}
I(X;Y) \leq -\frac{n}{2}\log (\kappa_{\ell}(\varepsilon,\delta)) \le \varepsilon_{\rm MI}.
\end{align*}
Thus, \eqref{eq:Gaussian} is $\varepsilon_{\rm MI}$-mutual-information private.

(Proof of (ii))
If \eqref{eq:Gaussian} is $\varepsilon_{\rm MI}$-mutual-information private, i.e.,~\eqref{eq:mip} holds, it follows from~\eqref{eq:kappa} and~\eqref{eq:mi} that
\begin{align*}
- 2 I(X;Y) 
\geq - 2\varepsilon_{\rm MI} \geq  n \log (\kappa_{\ell}(\varepsilon, \delta)).
\end{align*}
According to Lemma~\ref{lem:mi}, this implies that \eqref{eq:Gaussian} is $(\varepsilon,\delta)$-PML private.
\QED

\section{Proof of Theorem~\ref{thm:kalman_lb}}
\label{app:kalman}

(Step 1)
Applying Theorem~\ref{thm:privacy_con} to the Gaussian mechanism~\eqref{eq:Gaussian} gives the following necessary and sufficient condition for its \((\varepsilon, \delta)\)-PML privacy:
\begin{align*}
   &\log \det(\Sigma_{XX} -\Sigma_{XX} C^\top (C \Sigma_{XX} C^\top + \Theta)^{-1} C \Sigma_{XX})\\
   &\ge \fF^{-1}_{\chi^2_l}(1-\delta) - 2\varepsilon + \log \det(\Sigma_{XX}).
\end{align*}
Thus, it suffices to show
\begin{align}\label{pf1:kalman_lb}
&\log \det (P) - \log \det (Q) \nonumber\\
&-\log \det(\Sigma_{XX} -\Sigma_{XX} C^\top (C \Sigma_{XX} C^\top + \Theta)^{-1} C \Sigma_{XX}) \nonumber\\
&+ \log \det(\Sigma_{XX}) \ge 0.
\end{align}

(Step 2)
We rearrange~\eqref{pf1:kalman_lb}. First, the matrix inversion lemma gives
\begin{align*}
&\Sigma_{XX} -\Sigma_{XX} C^\top(C \Sigma_{XX} C^\top + \Theta)^{-1} C \Sigma_{XX}\\
&=(\Sigma_{XX}^{-1} + C^\top \Theta^{-1} C)^{-1}.
\end{align*}

Next, let $P^- \in \bS_{++}^n$ be the (unique) solution to the following discrete-time algebraic Riccati equation:
\begin{align}\label{eq:Ric}
P^- = A P^- A^{\top} + Q - A P^- C^\top (C P^- C^\top + \Theta)^{-1} C P^- A^\top.
\end{align}
Then, the steady-state covariance $P$ satisfies 
\begin{align*}
P 
&= P^- - P^- C^\top (C P^- C^\top + \Theta )^{-1} C P^-\\
&= ((P^-)^{-1} + C^\top \Theta^{-1} C)^{-1}.
\end{align*}

Thus, \eqref{pf1:kalman_lb} is equivalent to
\begin{align}\label{pf2:kalman_lb}
&\log \det(\Sigma_{XX}^{-1} + C^\top \Theta^{-1} C) + \log \det(\Sigma_{XX}) \nonumber\\
&\ge \log \det ((P^-)^{-1} + C^\top \Theta^{-1} C ) + \log \det (Q).
\end{align}

(Step 3)
We show~\eqref{pf2:kalman_lb}. Subtracting~\eqref{eq:Ric} from~\eqref{eq:Lya} gives the following Lyapunov equation  
\begin{align*}
(\Sigma_{XX}-P^-) 
&= A (\Sigma_{XX}-P^-) A^{\top} \\
&\quad +  A P^- C^\top (C P^- C^{\top} + \Theta)^{-1} C P^- A^\top.
\end{align*}
Since $A$ is Schur stable and $A P^- C^\top (C P^- C^{\top} + \Theta)^{-1} C P^- A^\top\succeq 0$, we have $\Sigma_{XX}\succeq P^-$.
This implies
\begin{align*}
&\log \det(\Sigma_{XX}^{-1} + C^\top \Theta^{-1} C) + \log \det(\Sigma_{XX})\\
&=\log \det(I + \Sigma_{XX}^{\frac{1}{2}}C^\top \Theta^{-1} C \Sigma_{XX}^{\frac{1}{2}})\\
&\ge \log \det(I + (P^-)^{\frac{1}{2}}C^\top \Theta^{-1} C (P^-)^{\frac{1}{2}})\\
&=  \log \det((P^-)^{-1} + C^\top \Theta^{-1} C) + \log \det((P^-)).
\end{align*}

Next, by usng the matrix inversion lemma, \eqref{eq:Ric} can be rearranged as
\begin{align*}
P^- = A ((P^-)^{-1} + C^\top \Theta^{-1} C)^{-1} A^{\top} + Q,
\end{align*}
which implies $\log \det((P^-)) \ge \log \det(Q)$. Therefore, we have~\eqref{pf2:kalman_lb}.
\QED






\bibliographystyle{ieeetr}
\bibliography{ref}

@ARTICLE{Yu2020,
  author={Kawano, Yu and Cao, Ming},
  journal={IEEE Transactions on Automatic Control}, 
  title={Design of Privacy-Preserving Dynamic Controllers}, 
  year={2020},
  volume={65},
  number={9},
  pages={3863-3878},
  doi={10.1109/TAC.2020.2994030}}

@article{Yu2021,
title = {Modular control under privacy protection: Fundamental trade-offs},
journal = {Automatica},
volume = {127},
pages = {109518},
year = {2021},
issn = {0005-1098},
doi = {https://doi.org/10.1016/j.automatica.2021.109518},
url = {https://www.sciencedirect.com/science/article/pii/S0005109821000388},
author = {Yu Kawano and Kenji Kashima and Ming Cao}
}

@article{farokhi2019ensuring,
  title={Ensuring privacy with constrained additive noise by minimizing fisher information},
  author={Farokhi, Farhad and Sandberg, Henrik},
  journal={Automatica},
  volume={99},
  pages={275--288},
  year={2019},
  publisher={Elsevier}
}

@inproceedings{altafini2019dynamical,
  title={A dynamical approach to privacy preserving average consensus},
  author={Altafini, Claudio},
  booktitle={2019 IEEE 58th Conference on Decision and Control},
  pages={4501--4506},
  year={2019},
  organization={IEEE}
}

@article{wang2019privacy,
  title={Privacy-preserving average consensus via state decomposition},
  author={Wang, Yongqiang},
  journal={IEEE Transactions on Automatic Control},
  volume={64},
  number={11},
  pages={4711--4716},
  year={2019},
  publisher={IEEE}
}

@ARTICLE{Boyd2008,
  author={Joshi, Siddharth and Boyd, Stephen},
  journal={IEEE Transactions on Signal Processing}, 
  title={Sensor Selection via Convex Optimization}, 
  year={2009},
  volume={57},
  number={2},
  pages={451-462},
  doi={10.1109/TSP.2008.2007095}}

@article{wang2023differential,
  title={Differential initial-value privacy and observability of linear dynamical systems},
  author={Wang, Lei and Manchester, Ian R and Trumpf, Jochen and Shi, Guodong},
  journal={Automatica},
  volume={148},
  pages={110722},
  year={2023},
  publisher={Elsevier}
}

@inproceedings{saeidian2023pointwise,
  title={Pointwise maximal leakage on general alphabets},
  author={Saeidian, Sara and Cervia, Giulia and Oechtering, Tobias J and Skoglund, Mikael},
  booktitle={2023 IEEE International Symposium on Information Theory},
  pages={388--393},
  year={2023},
  organization={IEEE}
}

@article{saeidian2023pointwise_tit,
  author={Saeidian, Sara and Cervia, Giulia and Oechtering, Tobias J. and Skoglund, Mikael},
  journal={IEEE Transactions on Information Theory}, 
  title={Pointwise Maximal Leakage}, 
  year={2023},
  volume={69},
  number={12},
  pages={8054-8080},
  doi={10.1109/TIT.2023.3304378}
}

@inproceedings{balle2018improving,
  title={Improving the gaussian mechanism for differential privacy: Analytical calibration and optimal denoising},
  author={Balle, Borja and Wang, Yu-Xiang},
  booktitle={International Conference on Machine Learning},
  pages={394--403},
  year={2018},
  organization={PMLR}
}

@inproceedings{dwork2006differential,
  title={Differential privacy},
  author={Dwork, Cynthia},
  booktitle={International Colloquium on Automata, Languages, and Programming},
  pages={1--12},
  year={2006},
  organization={Springer}
}

@article{wang2016relation,
  title={On the relation between identifiability, differential privacy, and mutual-information privacy},
  author={Wang, Weina and Ying, Lei and Zhang, Junshan},
  journal={IEEE Transactions on Information Theory},
  volume={62},
  number={9},
  pages={5018--5029},
  year={2016},
  publisher={IEEE}
}

@article{jiang2021context,
  title={Context-aware local information privacy},
  author={Jiang, Bo and Seif, Mohamed and Tandon, Ravi and Li, Ming},
  journal={IEEE Transactions on Information Forensics and Security},
  volume={16},
  pages={3694--3708},
  year={2021},
  publisher={IEEE}
}

@inproceedings{dwork2006our,
  title={Our data, ourselves: Privacy via distributed noise generation},
  author={Dwork, Cynthia and Kenthapadi, Krishnaram and McSherry, Frank and Mironov, Ilya and Naor, Moni},
  booktitle={24th International Conference on the Theory and Applications of Cryptographic Techniques},
  pages={486--503},
  year={2006},
  organization={Springer}
}

@book{anderson2005optimal,
  title={Optimal filtering},
  author={Anderson, Brian DO and Moore, John B},
  year={2005},
  publisher={Courier Corporation}
}

@book{golub2013matrix,
  title={Matrix Computations},
  author={Golub, Gene H and Van Loan, Charles F},
  year={2013},
  publisher={JHU press}
}

@article{weng2025optimal,
  title={Optimal Privacy-Aware State Estimation},
  author={Weng, Chuanghong and Nekouei, Ehsan and Johansson, Karl H},
  journal={IEEE Transactions on Automatic Control},
  year={2025},
note    = {(early access)},
  publisher={IEEE}
}

@inproceedings{alisic2020ensuring,
  title={Ensuring privacy of occupancy changes in smart buildings},
  author={Alisic, Rijad and Molinari, Marco and Par{\'e}, Philip E and Sandberg, Henrik},
  booktitle={2020 IEEE Conference on Control Technology and Applications},
  pages={871--876},
  year={2020},
  organization={IEEE}
}

@article{han2018privacy,
  title={Privacy in control and dynamical systems},
  author={Han, Shuo and Pappas, George J},
  journal={Annual Review of Control, Robotics, and Autonomous Systems},
  volume={1},
  number={1},
  pages={309--332},
  year={2018},
  publisher={Annual Reviews}
}

@article{yu2016smart,
  title={Smart grids: A cyber--physical systems perspective},
  author={Yu, Xinghuo and Xue, Yusheng},
  journal={Proceedings of the IEEE},
  volume={104},
  number={5},
  pages={1058--1070},
  year={2016},
  publisher={IEEE}
}

@article{venancio2023cps,
  title={How CPS and Autonomous Robots are Integrated to other I4. 0 Technologies: a systematic literature review},
  author={Venancio Teixeira, J{\^o}natas and da Silva Hounsell, Marcelo and Wildgrube Bertol, Douglas},
  journal={Production \& Manufacturing Research},
  volume={11},
  number={1},
  pages={2279715},
  year={2023},
  publisher={Taylor \& Francis}
}

@article{sisinni2018industrial,
  title={Industrial internet of things: Challenges, opportunities, and directions},
  author={Sisinni, Emiliano and Saifullah, Abusayeed and Han, Song and Jennehag, Ulf and Gidlund, Mikael},
  journal={IEEE Transactions on Industrial Informatics},
  volume={14},
  number={11},
  pages={4724--4734},
  year={2018},
  publisher={IEEE}
}

@inproceedings{koufogiannis2017differential,
  title={Differential privacy for dynamical sensitive data},
  author={Koufogiannis, Fragkiskos and Pappas, George J},
  booktitle={56th IEEE  Conference on Decision and Control},
  pages={1118--1125},
  year={2017},
  organization={IEEE}
}

@article{dong2022gaussian,
  title={Gaussian differential privacy},
  author={Dong, Jinshuo and Roth, Aaron and Su, Weijie J},
  journal={Journal of the Royal Statistical Society Series B: Statistical Methodology},
  volume={84},
  number={1},
  pages={3--37},
  year={2022},
  publisher={Oxford University Press}
}

@article{dwork2016concentrated,
  title={Concentrated differential privacy},
  author={Dwork, Cynthia and Rothblum, Guy N},
  journal={arXiv preprint arXiv:1603.01887},
  year={2016}
}

@inproceedings{mironov2017renyi,
  title={R{\'e}nyi differential privacy},
  author={Mironov, Ilya},
  booktitle={IEEE 30th Computer Security Foundations Symposium},
  pages={263--275},
  year={2017},
  organization={IEEE}
}

@article{zhang2023age,
  title={Age-dependent differential privacy},
  author={Zhang, Meng and Wei, Ermin and Berry, Randall and Huang, Jianwei},
  journal={IEEE Transactions on Information Theory},
  volume={70},
  number={2},
  pages={1300--1319},
  year={2023},
  publisher={IEEE}
}

@article{sweeney2002k,
  title={k-anonymity: A model for protecting privacy},
  author={Sweeney, Latanya},
  journal={International Journal of Uncertainty, Fuzziness and Knowledge-based systems},
  volume={10},
  number={05},
  pages={557--570},
  year={2002},
  publisher={World Scientific}
}

@article{machanavajjhala2007diversity,
  title={l-diversity: Privacy beyond k-anonymity},
  author={Machanavajjhala, Ashwin and Kifer, Daniel and Gehrke, Johannes and Venkitasubramaniam, Muthuramakrishnan},
  journal={Acm Transactions on Knowledge Discovery from Data},
  volume={1},
  number={1},
  pages={3--es},
  year={2007},
  publisher={ACM New York, NY, USA}
}

@inproceedings{li2006t,
  title={t-closeness: Privacy beyond k-anonymity and l-diversity},
  author={Li, Ninghui and Li, Tiancheng and Venkatasubramanian, Suresh},
  booktitle={23rd  IEEE International Conference on Data Engineering},
  pages={106--115},
  year={2006},
  organization={IEEE}
}

@inproceedings{zhao2014achieving,
  title={Achieving differential privacy of data disclosure in the smart grid},
  author={Zhao, Jing and Jung, Taeho and Wang, Yu and Li, Xiangyang},
  booktitle={IEEE Conference on Computer Communications},
  pages={504--512},
  year={2014},
  organization={IEEE}
}

@article{sivakumar2024addressing,
  title={Addressing privacy concerns with wearable health monitoring technology},
  author={Sivakumar, CLV and Mone, Varda and Abdumukhtor, Rakhmanov},
  journal={Wiley Interdisciplinary Reviews: Data Mining and Knowledge Discovery},
  volume={14},
  number={3},
  pages={e1535},
  year={2024},
  publisher={Wiley Online Library}
}

@article{ITO2021109732,
title = {Privacy protection with heavy-tailed noise for linear dynamical systems},
journal = {Automatica},
volume = {131},
pages = {109732},
year = {2021},
issn = {0005-1098},
doi = {https://doi.org/10.1016/j.automatica.2021.109732},
url = {https://www.sciencedirect.com/science/article/pii/S0005109821002521},
author = {Kaito Ito and Yu Kawano and Kenji Kashima},
}

@article{WKW:25,
author = {Watanabe, Rintaro and Kawano, Yu and Wada, Nobutaka and Cao, Ming},
title = {Frequency Shaping for Improving a Trade-Off Between Control and Privacy Performance: Beyond Differential Privacy},
journal = {International Journal of Robust and Nonlinear Control},
volume = {},
number = {},
pages = {},
year = {2025},
doi = {https://doi.org/10.1002/rnc.7789},
note = {(early access)},
}

@article{zhang2019security,
  title={Security and privacy on blockchain},
  author={Zhang, Rui and Xue, Rui and Liu, Ling},
  journal={ACM Computing Surveys},
  volume={52},
  number={3},
  pages={1--34},
  year={2019},
  publisher={ACM New York, NY, USA}
}

@inproceedings{dwork2006calibrating,
  title={Calibrating noise to sensitivity in private data analysis},
  author={Dwork, Cynthia and McSherry, Frank and Nissim, Kobbi and Smith, Adam},
  booktitle={2006 Theory of Cryptography Conference},
  pages={265--284},
  year={2006},
  organization={Springer}
}

@article{nozari2017differentially,
  title={Differentially private average consensus: Obstructions, trade-offs, and optimal algorithm design},
  author={Nozari, Erfan and Tallapragada, Pavankumar and Cort{\'e}s, Jorge},
  journal={Automatica},
  volume={81},
  pages={221--231},
  year={2017},
  publisher={Elsevier}
}

@article{hosseinalizadeh2024preserving,
  title={Preserving Privacy in Cloud-based Data-Driven Stabilization},
  author={Hosseinalizadeh, Teimour and Monshizadeh, Nima},
  journal={arXiv preprint arXiv:2410.17353},
  year={2024}
}

@article{liu2024design,
  author={Liu, Le and Kawano, Yu and Cao, Ming},
  journal={IEEE Transactions on Automatic Control}, 
  title={Design of Stochastic Quantizers for Privacy-Preserving Control}, 
  year={2025},
  volume={},
  number={},
  pages={1-15},
note ={(early access)},
  doi={10.1109/TAC.2025.3608135}
}

@article{le2013differentially,
  title={Differentially private filtering},
  author={Le Ny, Jerome and Pappas, George J},
  journal={IEEE Transactions on Automatic Control},
  volume={59},
  number={2},
  pages={341--354},
  year={2013},
  publisher={IEEE}
}

@article{wang2024robust,
  title={Robust constrained consensus and inequality-constrained distributed optimization with guaranteed differential privacy and accurate convergence},
  author={Wang, Yongqiang and Nedi{\'c}, Angelia},
  journal={IEEE Transactions on Automatic Control},
  volume={69},
  number={11},
  pages={7463--7478},
  year={2024},
  publisher={IEEE}
}

@article{wang2023tailoring,
  title={Tailoring gradient methods for differentially private distributed optimization},
  author={Wang, Yongqiang and Nedi{\'c}, Angelia},
  journal={IEEE Transactions on Automatic Control},
  volume={69},
  number={2},
  pages={872--887},
  year={2023},
  publisher={IEEE}
}

@article{chen2024local,
  author={Chen, Ziqin and Wang, Yongqiang},
  journal={IEEE Transactions on Automatic Control}, 
  title={Local Differential Privacy for Decentralized Online Stochastic Optimization With Guaranteed Optimality and Convergence Speed}, 
  year={2025},
  volume={70},
  number={7},
  pages={4238-4253},
  doi={10.1109/TAC.2024.3519938}
}

@article{moradi2022privacy,
  title={Privacy-preserving distributed Kalman filtering},
  author={Moradi, Ashkan and Venkategowda, Naveen KD and Talebi, Sayed Pouria and Werner, Stefan},
  journal={IEEE Transactions on Signal Processing},
  volume={70},
  pages={3074--3089},
  year={2022},
  publisher={IEEE}
}

@article{liu2025initial,
  title={Initial State Privacy of Nonlinear Systems on Riemannian Manifolds},
  author={Liu, Le and Kawano, Yu and Xie, Antai and Cao, Ming},
  journal={International Journal of Robust and Nonlinear Control},
  year={2025},
  publisher={Wiley Online Library}
}

@article{yazdani2022differentially,
  title={Differentially private {LQ} control},
  author={Yazdani, Kasra and Jones, Austin and Leahy, Kevin and Hale, Matthew},
  journal={IEEE Transactions on Automatic Control},
  volume={68},
  number={2},
  pages={1061--1068},
  year={2022},
  publisher={IEEE}
}

@article{nekouei2022model,
  title={A model randomization approach to statistical parameter privacy},
  author={Nekouei, Ehsan and Sandberg, Henrik and Skoglund, Mikael and Johansson, Karl Henrik},
  journal={IEEE Transactions on Automatic Control},
  volume={68},
  number={2},
  pages={839--850},
  year={2022},
  publisher={IEEE}
}

@article{murguia2021privacy,
  title={On privacy of dynamical systems: An optimal probabilistic mapping approach},
  author={Murguia, Carlos and Shames, Iman and Farokhi, Farhad and Ne{\v{s}}i{\'c}, Dragan and Poor, H Vincent},
  journal={IEEE Transactions on Information Forensics and Security},
  volume={16},
  pages={2608--2620},
  year={2021},
  publisher={IEEE}
}

@article{grosse2024extremal,
    author={Grosse, Leonhard and Saeidian, Sara and Oechtering, Tobias J.},
  journal={IEEE Transactions on Information Forensics and Security}, 
  title={Extremal Mechanisms for Pointwise Maximal Leakage}, 
  year={2024},
  volume={19},
  number={},
  pages={7952-7967},
  doi={10.1109/TIFS.2024.3449556}}

@book{durrett2019probability,
  title={Probability: theory and examples},
  author={Durrett, Rick},
  volume={49},
  year={2019},
  publisher={Cambridge university press}
}

@book{petersen2008matrix,
  title={The Matrix Cookbook},
  author={Petersen, Kaare Brandt and Pedersen, Michael Syskind and others},
  journal={Technical University of Denmark},
  volume={7},
  number={15},
  pages={510},
  year={2008}
}

@book{bartle2000introduction,
  title={Introduction to Real Analysis},
  author={Bartle, Robert G and Sherbert, Donald R},
  volume={2},
  year={2000},
  publisher={Wiley New York}
}

@article{oldewurtel2012use,
  title={Use of model predictive control and weather forecasts for energy efficient building climate control},
  author={Oldewurtel, Frauke and Parisio, Alessandra and Jones, Colin N and Gyalistras, Dimitrios and Gwerder, Markus and Stauch, Vanessa and Lehmann, Beat and Morari, Manfred},
  journal={Energy and buildings},
  volume={45},
  pages={15--27},
  year={2012},
  publisher={Elsevier}
}

@article{oldewurtel2013stochastic,
  title={Stochastic model predictive control for building climate control},
  author={Oldewurtel, Frauke and Jones, Colin Neil and Parisio, Alessandra and Morari, Manfred},
  journal={IEEE Transactions on Control Systems Technology},
  volume={22},
  number={3},
  pages={1198--1205},
  year={2013},
  publisher={IEEE}
}

@book{horn2012matrix,
  title={Matrix Analysis},
  author={Horn, Roger A and Johnson, Charles R},
  year={2012},
  publisher={Cambridge university press}
}

@book{cover1999elements,
  title={Elements of information theory},
  author={Cover, Thomas M},
  year={1999},
  publisher={John Wiley \& Sons}
}

@article{ke2025privacy,
  title={Privacy-Preserving Distributed Estimation with Limited Data Rate},
  author={Ke, Jieming and Wang, Jimin and Zhang, Ji-Feng},
  journal={arXiv preprint arXiv:2510.12549},
  year={2025}
}

\begin{IEEEbiography}{Le Liu}
received the Bachelor of Engineering degree in power and energy engineering and Master of Engineering degree in systems and control
from Northwest A\&F University and Dalian University of Technology, China, in 2019 and 2022, respectively. He is currently working toward the
Ph.D. Degree in systems and control in Faculty of Science and Engineering from
University of Groningen, The Netherlands.
His research interests include privacy of control systems, networked control systems and nonlinear systems.
\end{IEEEbiography}

\begin{IEEEbiography}{Yu Kawano}(M'13) 
has since 2019 been an Associate Professor in the Graduate School of Advanced Science and Engineering at Hiroshima University. He received the M.S. and Ph.D. degrees in Engineering from Osaka University, Japan, in 2011 and 2013, respectively. As a Post-Doctoral Researcher, he then joined Kyoto University, Japan and moved in 2016 to the University of Groningen, The Netherlands. He has held visiting research positions at Tallinn University of Technology, Estonia, the University of Groningen, the University of Pavia, Italy, and the Indian Institute of Technology Bombay, India. His research interests include nonlinear systems, complex networks, model reduction, and privacy of control systems. He is an Associate Editor for Systems and Control Letters, IEEE CSS Conference Editorial Board, and EUCA Conference Editorial Board.
\end{IEEEbiography}

\begin{IEEEbiography}{Ming Cao} has since 2016 been a professor of networks and robotics with the Engineering and Technology Institute (ENTEG) at the University of Groningen, the Netherlands, where he started as an assistant professor in 2008. Since 2022 he is the director of the Jantina Tammes School of Digital Society, Technology and AI at the same university. He received the Bachelor degree in 1999 and the Master degree in 2002 from Tsinghua University, China, and the Ph.D. degree in 2007 from Yale University, USA, all in electrical engineering. From 2007 to 2008, he was a Research Associate at Princeton University, USA. He worked as a research intern in 2006 at the IBM T. J. Watson Research Center, USA. He is the 2017 and inaugural recipient of the Manfred Thoma medal from the International Federation of Automatic Control (IFAC) and the 2016 recipient of the European Control Award sponsored by the European Control Association (EUCA). He is an IEEE fellow. He is a Senior Editor for Systems and Control Letters, an Associate Editor for IEEE Transactions on Automatic Control, IEEE Transaction of Control of Network Systems and IEEE Robotics \& Automation Magazine, and was an associate editor for IEEE Transactions on Circuits and Systems and IEEE Circuits and Systems Magazine. He is a member of the IFAC Council and a vice chair of the IFAC Technical Committee on Large-Scale Complex Systems. His research interests include autonomous robots and multi-agent systems, complex networks and decision-making processes. 
\end{IEEEbiography}

\end{document}